\documentclass[authoryear]{elsarticle}       
\usepackage{graphicx}
\usepackage{xcolor}
\usepackage{amsfonts,amsmath,amsthm}
\usepackage{natbib}

\newtheorem{theorem}{Theorem}[section]
\newtheorem{lemma}[theorem]{Lemma}
\newtheorem{corollary}[theorem]{Corollary}
\newtheorem{proposition}[theorem]{Proposition}
\theoremstyle{definition}
\newtheorem{definition}[theorem]{Definition}
\newtheorem{example}[theorem]{Example}
\newtheorem{remark}[theorem]{Remark}

\newcommand{\R}{\mathbb{R}}
\newcommand{\Z}{\mathbb{Z}}
\newcommand{\Q}{\mathbb{Q}}
\newcommand{\N}{\mathbb{N}}
\newcommand{\CS}{\mathit{CS}}
\newcommand{\calCS}{{\mathcal{CS}}}

\newcommand{\tup}[1]{\langle #1\rangle }
\newcommand{\supp}{{\mathit{supp}}}

\renewcommand{\vec}[1]{\mathbf{#1}}
\renewcommand{\cal}[1]{\mathcal{#1}}

\newcommand{\optval}{\mbox{\textsc{OptVal}}}
\newcommand{\arbval}{\mbox{\textsc{ArbVal}}}

\newcommand{\indset}{\mbox{\textsc{Independent-Set}}}
\newcommand{\setcover}{\mbox{\textsc{Set-Cover}}}
\newcommand{\checkcore}{\mbox{\textsc{CheckCore}}}
\newcommand{\isstable}{\mbox{\textsc{Is-Stable}}}
\newcommand{\wmax}{W_{M}}
\newcommand{\eps}{\varepsilon}
\newcommand{\NP}{\mathit{NP}}
\newcommand{\poly}{\mathrm{poly}}
\newcommand{\width}{\mathrm{width}}
\newcommand{\tw}{\mathrm{tw}}
\newcommand{\eqdef}{\stackrel{\mathrm{def}}{=}}
\newcommand{\opt}{\mathrm{opt}}
\newcommand{\nodes}{\mathrm{N}}
\newcount\Comments  
\newcommand{\kibitz}[2]{\ifnum\Comments=1{\color{#1}{#2}}\fi}			
\newcommand{\yz}[1]{\kibitz{red}{[Yair:#1]}}							
\newcommand{\gc}[1]{\kibitz{purple}{[George:#1]}}						
\date{Received: date / Accepted: date}

\begin{document}
\begin{frontmatter}

\title{Cooperative Games with Overlapping Coalitions: Charting the Tractability Frontier}

\author[cmu]{Yair Zick}
\ead{yairzick@cmu.edu}
\author[tuc]{Georgios Chalkiadakis}
\ead{gehalk@ece.tuc.gr}
\author[oxf]{Edith Elkind}
\ead{elkind@cs.ox.ac.uk}
\author[aueb]{Evangelos Markakis}
\ead{markakis@gmail.com}
\address[cmu]{School of Computer Science, 
							Carnegie-Mellon University, USA}
\address[tuc]{School of Electronic and Computer Engineering, 
							Technical University of Crete, Greece}
\address[oxf]{Department of Computer Science, University of Oxford, UK}
\address[aueb]{Department of Informatics, 
							Athens University of Economics and Business, Greece}

\begin{abstract}
In many multiagent scenarios, agents distribute resources, such as time or energy, among several tasks.
Having completed their tasks and generated profits, task payoffs must be divided among the agents in some reasonable manner.
Cooperative games with overlapping coalitions (OCF games) are a recent framework proposed by \cite{ocfgeb}, generalizing classic cooperative games to the case where agents may belong to more than one coalition. Having formed overlapping coalitions and divided profits, some agents may feel dissatisfied with their share of the profits, and would like to deviate from the given outcome. However, deviation in OCF games is a complicated matter: agents may decide to withdraw only {\em some} of their weight from {\em some} of the coalitions they belong to; that is, even after deviation, it is possible that agents will still be involved in tasks with non-deviators. This means that the desirability of a deviation, and the stability of formed coalitions, 
is to a great extent determined by the reaction of non-deviators.
In this work, we explore algorithmic aspects of OCF games, focusing on the {\em core} in OCF games. We study the problem of deciding if the core of an OCF game is not empty, and whether a core payoff division can be found in polynomial time;
moreover, we identify conditions that ensure that the problem admits polynomial time algorithms. 
Finally, we introduce and study a natural class of OCF games, the {\em Linear Bottleneck Games}. Interestingly, we show that such games always have a non-empty core, even assuming a highly lenient reaction to deviations.

\end{abstract}

\begin{keyword}
Cooperative Games \sep Overlapping Coalition Formation \sep Core \sep Treewidth \sep Arbitration Functions
\end{keyword}
\end{frontmatter}
\section{Introduction}
Consider a simple market exchange between several vendors, each owning a single, divisible good. Vendors perform pairwise exchanges (say, oil for sugar) and may charge different prices to different buyers. Suppose that having agreed on a transaction schedule and payments, one vendor decides that he is unhappy with the amount of money he receives from the transaction with buyer 1, and wishes to cancel the deal. However, buyer 2, upon hearing that the transaction with buyer 1 has been canceled, no longer wishes to work with the seller, and cancels the agreement with him as well. 

This setting features several interesting characteristics. First, a vendor may allocate resources to several buyers, and may also buy from several vendors. In other words, vendors may allocate resources to several profit-generating tasks. Second, agents may withdraw some of their resources from some agreements. For example, an oil vendor may wish to sell less oil to some client, but not change his interactions with other parties; similarly, a vendor may want to pay less to some sellers, but maintain the same payments to others. Finally, when trying to strategically change an agreement, vendors must be aware of how their actions affect the contracts they still maintain with other (possibly unaffected) parties. 

In this setting, agents must collaborate (exchange goods) in order to generate revenue. Having generated revenue by making an exchange, agents are free to share the profits of the exchange as they see fit. Profit sharing can be done directly, if the exchange of goods results in profit, e.g., if the agents can produce a new good using their resources and sell it for a profit; it can also be indirect, e.g., via setting a price for the good sold. To conclude, agents must first work together to generate profits, and, subsequently, share profits among themselves in some reasonable manner. When sharing profits, agents should account for individuals or groups of agents who are underpaid. A group of agents who can get more money by deviating from the proposed agreement may destabilize the entire market, resulting in a cascade of deviations that may result in a less desirable state (not to mention the cost of actual deviations). However, what constitutes a profitable deviation greatly depends on how non-deviators respond to a deviating set. 

Modeling this system of incentives and reactions is a challenge in itself, which has only recently been addressed. In their paper, \cite{ocfgeb} propose a novel approach to modeling scenarios where agents can divide resources among several coalitions, called {\em overlapping coalition formation games} (OCF games). Following their work, \cite{zick2014jair} propose a general model for handling deviations in such settings. These two works offer a rather comprehensive conceptual model for handling deviation behavior and analyzing stability in settings where agents work on several concurrent projects, and may only partially deviate. 

While conceptual models of strategic interactions in OCF games have been developed to a rather satisfactory degree, there has been little work on computing solution concepts for such games. The theory of OCF games is a generalization of classic cooperative game theory~\citep{coopbook}, where computational issues are relatively well-understood (see~\cite{compcoopbook} for an overview). However, as \cite{ocfgeb} show, more elaborate reactions to deviation may make even tractable instances of computing solution concepts for classic cooperative games NP-hard. For example,~\cite{ocfgeb} consider a class of games called {\em threshold task games}, a class generalizing {\em weighted voting games}). It is shown that finding a payoff division in the core can be done in pseudo-polynomial time (i.e., in time polynomial in the number of agents and in the number of bits used to encode the largest weight of any agent), if one assumes that when a set of agents deviates, no other agent will want to work with it---a reaction termed {\em conservative}. However, if non-deviators agree to work with a deviating set in any coalition that is unhurt by the deviation, a reaction termed {\em refined}, the same problem becomes NP-hard. 

To conclude, when assessing the computational complexity of finding stable outcomes in OCF games, one needs to consider not only structural properties of the value function (i.e., the way that agents generate profits), but also the way agents react to deviations. In our work, we study some of the computational issues that arise from these two aspects of OCF games, when deciding stability related questions in OCF games. 
\subsection{Our Contribution}
In this work we study several closely related questions:
\begin{enumerate}
	\item Given an OCF game, find an optimal coalition structure---i.e., an optimal way for agents to divide into groups and generate profits.
	\item Given an OCF game, a coalition structure, a payoff division and a subset of agents, compute the most that these agents can get by deviating.
	\item Given an OCF game, a coalition structure and a payoff division, decide whether the given payoff division is stable, i.e, does there exist a subset of agents that can receive more by deviating, given the reaction non-deviators would have to the deviation?
	\item Given an OCF game, find a stable outcome if one exists---i.e., find a coalition structure and a payoff division that ensures that no subset of agents can deviate, again given the reaction of the rest of the agents to the deviation.
\end{enumerate}
Our paper makes two attempts to answer these questions. In the first part of this paper (sections~\ref{sec:optval} to \ref{sec:tw}), we make no assumptions on how agents generate revenue (i.e., the characteristic function can take any form), and focus instead on the {\em structure} of their interaction. 
In Section~\ref{sec:optval}, we show that in order to circumvent computational intractability, several conditions must be met. First, given the results of \cite{ocfgeb}, we search for pseudopolynomial time algorithms that decide the above questions (foregoing this assumption results in NP-hardness, even for a single agent). Second, agents cannot form large coalitions; third, as shown in Section~\ref{sec:arbval}, agent reaction to deviators must also be limited in its scope, as complex reactions to deviation are a source of computational complexity in themselves. 
Inspired by the work of \cite{Demange2004}, we study settings where agent
interactions have a tree structure. Under these assumptions, optimal coalition 
structures and profitable deviations can be found (as shown in Sections~\ref{sec:optval} and~\ref{sec:arbval}); moreover, deciding whether a payoff division is 
stable and whether the core is not empty can be done in polynomial time (see Section~\ref{sec:stable}).

In order to provide a clearer picture of the role of interaction structure in limiting computational complexity, we provide a connection between the treewidth of the {\em agent 
interaction graph}, a graph whose edges correspond to allowable agent interactions, and the computational complexity of finding stable outcomes in OCF 
games. In Section~\ref{sec:tw}, we show how the algorithms described in previous 
sections can be extended to algorithms whose running time is polynomial in $n$ 
and $(\wmax)^k$, where $n$ is the number of agents, $\wmax$ is the maximal weight of any agent, and $k$ is the treewidth of the interaction graph. 

In the second part of our paper (Section~\ref{sec:LBG}), we make no assumptions on the structure of agent interactions, but rather limit our attention to a specific class of OCF games. That is, instead of limiting how agents interact, we limit the way in which they generate profits. More specifically, we study a class of games we term {\em linear bottleneck games} (LBGs); our main insight is that these games are stable even if one assumes a very lenient reaction to deviation. Our results extend those of \cite{markakis2005core}, who show that a subclass of these games (namely, multicommodity flow games) is stable in the classic cooperative sense. 

\section{Preliminaries}
In what follows, we refer to sets using uppercase letters, and to vectors in boldface letters. Given two vectors $\vec x,\vec y \in \R^m$, we say that $\vec x \le \vec y$ if for all $j \in \{1,\dots,m\}$, we have that $x^i \le y^i$. Given a set of agents $S \subseteq \{1,\dots,n\}$ and a vector $\vec x \in \R^m$ we denote by $\vec x^S$ the vector that equals $\vec x$ on all coordinates $i \in S$, and is 0 elsewhere. For a set $S \subseteq \{1,\dots,n\}$ and a vector $\vec x \in \R^n$, we write $x(S) = \sum_{i \in S}x^i$. Finally, we let $\vec e^S$ be the indicator vector of $S$ in $\R^n$: the $i$-th coordinate of $\vec e^S$ is 1 if $i \in S$, and is 0 otherwise. 

We begin by recalling the definition of a classic {\em cooperative game}~\citep{coopbook}. A cooperative game $\cal G$ is a tuple $\tup{N,u}$, where $N = \{1,\dots,n\}$ is a set of agents, and $u:2^N \to \R_+$ is a function that assigns a value to every subset of agents $S \subseteq N$. Every subset $S \subseteq N$, will also 
be referred to as a {\em coalition}. A {\em coalition structure} is a partition of agents $\Pi = \{S_1,\dots,S_m\}$ into disjoint subsets; that is, for all $S_j,S_{j'} \in \Pi$, $S_j \cap S_{j'} = \emptyset$, and $\bigcup_{j = 1}^m S_j = N$. An {\em imputation} for $\Pi$ is a vector $\vec p = (p^1,\dots,p^n)$ that satisfies {\em individual rationality}: $p^i \ge u(\{i\})$ for all $i \in N$, and {\em coalitional efficiency}: for all $S_j \in \Pi$, $p(S_j) = u(S_j)$. In words, an imputation is a division of the payoffs generated by the agents forming $\Pi$ such that agents are individually incentivized to form a coalition structure, and that the payoff from forming a coalition $S_j$ is allocated only to the members of $S_j$. We observe that when forming a coalition structure, agents may not join more than one coalition. While this is a valid assumption in many multi-agent scenarios, it is often the case that agents split their resources among several projects, forming overlapping coalitions.

In what follows, we describe {\em discrete overlapping coalition formation (OCF) games}, a variant of {\em overlapping coalition formation} (OCF) games~\citep{ocfgeb}. 
A discrete OCF game is a tuple $\cal G = \tup{N,\vec W,v}$. The set $N$ is the set of agents as before; each agent $i \in N$ has a {\em weight} $W_i$, 
and these weights are collectively described by the vector $\vec W = (W_1,\dots,W_n)$. 
\gc{I did some changes here...incidentally, perhaps we should be using $w_i$ instead of $W_i$...}
\yz{I believe that $W_i$ is fine, since it is the absolute amount of resources that players have} 

One can think of $W_i$ as the amount of some resource that agent $i$ possesses. We write $\cal W = \{ \vec c \in \Z_+^n \mid \vec c \le \vec W\}$, which is simply the set $\{0,\dots,W_1\}\times\dots\times\{0,\dots,W_n\}$; the set $\cal W$ is the set of all possible ways in which agents can allocate resources to a {\em single} task. The {\em characteristic function} $v:\cal W \to \R_+$ receives as input a vector $\vec c \in \cal W$, and outputs a value $v(\vec c)$, describing the profit generated by the agents in $N$. In order to generate profits, each agent must decide how much of his resources will he contribute to a certain task; the contribution of agent $i$ is simply a number $c^i \in \{0,\dots,W_i\}$, where if $c^i = 0$ then agent $i$ contributes nothing to completing a task, and if $c^i = W_i$ then all of agent $i$'s resources are assigned to a task. A vector $\vec c \in \cal W$ is also called a {\em coalition}; $c^i$ refers to
{\em the contribution of agent $i$}, and the value $v(\vec c)$ is the total profit derived when agents contribute according to $\vec c$. Given a coalition $\vec c \in \cal W$, we define the {\em support} of $\vec c$ to be the set of all agents in $N$ that contribute some of their resources to $\vec c$. Formally, we write $\supp(\vec c) = \{i \in N \mid c^i >0\}$. Agents in $\supp(\vec c)$ are those agents that have some {\em vested interest} in $\vec c$: they are the only ones who may receive a share of the profits made by $\vec c$, and will potentially be affected by any changes to $\vec c$.

\begin{remark}\label{rem:discrete}
	In some departure from \cite{ocfgeb} and \cite{ocfarb}, 
	we focus on algorithmic aspects of finding stable payoff divisions in OCF games. In the general OCF setting, agents are not limited to integer contributions to coalitions, and may contribute any fraction of their resources to a coalition. 
	The generality of the ``classic'' OCF setting, however,
	introduces certain conceptual concerns---e.g., it allows agents to form an arbitrarily large number of different coalitions. 
	Thus, it is unclear how to represent the characteristic function in a general OCF game, as one must, in principle, be able to provide a value for infinitely many possible coalitions.

	 All definitions given here are discrete equivalents of the definitions given for arbitrated OCF games in~\citep{ocfarb}; in fact, the definitions given here can be used nearly verbatim to define non-discrete OCF games.
	 Indeed, if one defines a coalition as a vector $\vec c \in [0,1]^n$ and agent resources as values in $[0,1]$ then the classic OCF game definition is obtained, since
	 in that definition a coalition is a vector $\vec c \in [0,1]^n$, with the contribution value $c^i$ representing the fraction of resources that agent $i$ allocates to the coalition $\vec c$.
\end{remark}

It is often the case that in order to maximize their profits, agents would be interested in forming more than one coalition, and participate in a {\em coalition structure}. A coalition structure, $\CS$, is a finite list of coalitions, i.e., $\CS = (\vec c_1,\dots, \vec c_m)$; we let $|\CS|$ be the number of coalitions in $\CS$. Note that $\CS$ is a list (a multiset), rather than a set of coalitions; this is because it is possible that some coalitions may form more than once, having agents completing the same task several times. We use standard set notation to refer to elements of $\CS$. That is, if $\vec y \in \CS$, then $\vec y$ is one of the coalitions listed in $\CS$; if $\CS' \subseteq \CS$, then $\CS'$ is a list of coalitions that are all listed in $\CS$. Since no agent can contribute more than the total amount of resources he possesses, we require that $\sum_{\vec c \in \CS} c^i \le W_i$ for all $i \in N$. We overload notation and write $v(\CS)$ in order to refer to $\sum_{\vec c \in \CS} v(\vec c)$. Given a subset $S \subseteq N$, we denote by $\calCS(S)$ the set of all coalition structures that can be formed by members of $S$; that is, $\calCS(S)$ consists of all coalition structures $\CS$, such that for all $\vec c \in \CS$, $\supp(\vec c) \subseteq S$. For $S=N$, 
$\calCS(N)$ simply refers to $\calCS$. The {\em weight} of a coalition is written as $\vec w(\CS) = \sum_{ \vec c \in \CS} \vec c$. We say that $\CS$ is {\em efficient} if all agents in $S$ contribute all of their resources to $\CS$, that is, $\vec w(\CS) = \vec W^S$.

Given a coalition structure $\CS \in \calCS$ and a set $S \subseteq N$, the coalition structure $\CS$ {\em reduced to $S$}, $\CS|_S$, is defined to be $\left(\vec w \in \CS\mid \supp(\vec w) \subseteq S\right)$. The coalition structure $\CS|_S$ lists all coalitions in $\CS$ that are fully controlled by the members of $S$; if $S$ decides to deviate, the resources used in $\CS|_S$ are freely available to $S$, but any coalition $\vec w \in \CS\setminus \CS|_S$ contains non-$S$ members as well, hence any changes to such coalitions may lead to negative repercussions for $S$, should it choose to change its contribution to $\vec w$.

Given a game $\cal G = \tup{N,\vec W,v}$, we define the {\em superadditive cover} of $v$ to be the function $v^*:\cal W \to \R_+$, such that for all $\vec c \in \cal W$
$$v^*(\vec c) = \max\{v(\CS) \mid \vec w(\CS) = \vec c\}.$$
Simply put, $v^*(\vec c)$ is the maximal profit that agents can generate by forming coalitions when their total resources are given by $\vec c$. We note that in our model, for any $\vec c \in \cal W$ there always exists a coalition structure $\CS\in \calCS$ such that $\vec w(\CS) = \vec c$, and $v^*(\vec c) = v(\CS)$.

Finally, we observe that if $W_i = 1$ for all $i \in N$, the resulting game is a classic cooperative game with coalition structures~\citep{coalstruct}.

\subsection{Payoff Division}
Having formed a coalition structure, agents will then have to divide the profits generated by the coalitions formed. Given a coalition structure $\CS = (\vec c_1,\dots,\vec c_m)$, an {\em imputation} $\vec x = (\vec x_1,\dots,\vec x_m)$ is a list of $|\CS| = m$ vectors in $\R_+^n$. The vector $\vec x_j$ describes the way the profits from the coalition $\vec c_j$ are divided among the agents. Given a coalition $\vec c \in \CS$, and an imputation $\vec x$, we refer to the division of the profits from $\vec c$ as $\vec x(\vec c)$; the payoff to agent $i$ from $\vec c$ is $x^i(\vec c)$. 
Such profit divisions must satisfy the following rules:
\begin{description}
	\item[Coalitional efficiency:] for all $\vec c \in \CS$, the total payoff from $\vec x(\vec c)$ must equal $v(\vec c)$; that is, $\sum_{i = 1}^n x^i(\vec c)= v(\vec c)$.
	\item[No side payments:] for all $\vec c \in \CS$, if $i \notin \supp(\vec c)$, then $x^i(\vec c) = 0$. This condition simply means that agents that did not contribute any resources to the coalition $\vec c$ may not partake in the profits generated by $\vec c$.
	\item[Individual rationality:] for all $i \in N$, the total payoff to $i$ must exceed the most that $i$ can get by working alone, or $v^*(\vec W^i)$. Formally: $\sum_{\vec c \in \CS} x^i(\vec c) \ge v^*(\vec e^i)$ for all $i \in N$.
\end{description}
The set of all possible imputations for a coalition structure $\CS$ is denoted by $I(\CS)$; an {\em outcome} is a pair $(\CS,\vec x)$, where $\CS$ is a coalition structure and $\vec x$ is in $I(\CS)$. We also define $p_i(\CS,\vec x)$ to be the total payoff to agent $i$ under $(\CS,\vec x)$: $p_i(\CS,\vec x) = \sum_{\vec c \in \CS}^m x^i(\vec c)$; given some set $S \subseteq N$, we similarly define $p_S(\CS,\vec x)$ to be the total payoff to the set $S$ under $(\CS,\vec x)$: $p_S(\CS,\vec x) = \sum_{i \in S}p_i(\CS,\vec x)$. Note that our definition of an outcome in OCF games is a generalization of the notion of an outcome in classic cooperative games, where an outcome consists of a coalition structure $\Pi = \{S_1,\dots,S_m\}$, and an imputation $\vec p$ for $\Pi$.

\subsection{Deviation and Arbitration Functions}
In a classic cooperative game $\cal G = \tup{N,u}$, a subset of agents deviates from an outcome $(\Pi,\vec p)$ when its total payoff $p(S)$ is less than the profits it can generate on its own, $u(S)$. However, in OCF games, deviation is a much more complicated matter. The classic notion of deviation in cooperative games implicitly assumes that when a set of agents deviates, it may not retain any ties to non-deviators; that is, it measures the desirability of deviation against the most it can make on its own, and assumes that non-deviators will no longer collaborate with a deviating set. This is not necessarily the case in the OCF setting: when a set $S\subseteq N$ deviates from an outcome $(\CS,\vec x)$, it may still retain some resources in coalitions with agents in $N \setminus S$. Deciding whether to retain connections with non-deviators requires that $S$ has prior knowledge of the way non-deviators react to deviations. In other words, in order to assess the desirability of a deviation, $S$ must know how the non-members of $S$ will react to such a move. \cite{ocfgeb} were the first to point out this type of agent behavior, and have shown how different types of non-deviator reactions to deviation lead to different notions of stability. \cite{ocfarb} suggest a general framework for handling deviation in OCF games, which is termed {\em arbitration functions}.

We begin by formally defining a deviation in OCF games. Given an outcome $(\CS,\vec x)$ and an agent set $S$, a {\em deviation of $S$ from $(\CS,\vec x)$} is a coalition structure $\CS'$ whose coalitions describe the resources that $S$ withdraws from each coalition $\vec c \in \CS\setminus \CS|_S$. We let $\vec d_{\CS'}(\vec c)$ denote the coalition in $\CS'$ which describes the resources that $S$ withdraws from the coalition $\vec c$, omitting the $\CS'$ subscript when it is understood from context. Given a deviation $\CS'$, we require that $\vec d(\vec c)$ satisfies $\vec d(\vec c) \le \vec c$, and $\vec d(\vec c) \le \vec e^S$; the first requirement ensures that no agent in $S$ withdraws more resources than it has invested in $\vec c$, while the second requirement ensures that agents that do not belong to $S$ do not withdraw any resources from $\vec c$. Note that the formal definition of deviation does not refer to the imputation $\vec x$ in any explicit manner.

Having withdrawn resources from the coalition structure $\CS$, non-deviating agents ---the members of $N \setminus S$--- must decide how they react to the deviation. Formally, \cite{ocfarb} define the {\em arbitration function} $\cal A$ to be a function whose input is an outcome $(\CS,\vec x)$, a deviating set $S\subseteq N$, and a deviation $\CS'$ of $S$ from $\CS$. For each $\vec c \in \CS\setminus \CS|_S$, $\cal A$ outputs a value $\alpha_{\vec c}(\CS,\vec x,S,\CS')$, which specifies the total payoff that the coalition $\vec c$ offers $S$ as a result of its deviation. In order for a deviation to be profitable, members of $S$ must first form a coalition structure using all resources available to them, divide the payoffs from that coalition structure, and divide the (possibly negative) payoffs from $\alpha_{\vec c}$ such that every agent in $S$ receives a strictly higher payoff than what he receives under $(\CS,\vec x)$. If this is possible, we say that $\CS'$ is an {\em $\cal A$-profitable deviation} of $S$ from $(\CS,\vec x)$. An outcome $(\CS,\vec x)$ is called {\em $\cal A$-stable} if no subset $S \subseteq N$ can $\cal A$-profitably deviate from $(\CS,\vec x)$, and similarly, a game $\cal G$ is {\em $\cal A$-stable} if there exists an outcome $(\CS,\vec x)$ that is $\cal A$-stable. Suppose that $S$ forms the coalition structure $\CS''$ with the resources it has withdrawn from $\CS$ (and possibly with other unused resources from $S$); the total payoff that $S$ receives from the deviation $\CS'$ under $\cal A$, if it forms $\CS''$, is written as $\cal A(\CS,\vec x,S,\CS') = v(\CS'') + \sum_{ \vec c \in \CS\setminus \CS|_S} \alpha_{\vec c}(\CS,\vec x,S,\CS')$; the most that $S$ can get by deviating is denoted $$\cal A^*(\CS,\vec x,S) = \sup\{\cal A(\CS,\vec x,S,\CS') \mid \CS'\mbox{ is a deviation of } S \mbox{ from } (\CS,\vec x)\}.$$
\cite{ocfarb} provide a simple characterization of $\cal A$-stable games, given in the following theorem.
\begin{theorem}[\cite{ocfarb}]\label{thm:core}
	An OCF game $\cal G = \tup{N,v}$ is $\cal A$-stable if and only if there exists an outcome $(\CS,\vec x)$ such that for all $S \subseteq N$,
	$$p_S(\CS,\vec x) \ge \cal A^*(\CS,\vec x,S).$$
\end{theorem}
Theorem~\ref{thm:core} implies that, to verify stability,  there is no need to look for explicit deviations and check if all deviators are strictly better off. Instead, it is enough to simply ensure that the total payoff to $S$ is greater than the total payoff $S$ can receive by deviating---no matter how it deviates.

Before we proceed, let us describe some arbitration functions, which will also be the focus of this paper.\footnote{All arbitration functions mentioned here were introduced by \cite{ocfgeb}; however, \cite{ocfgeb} do not use the term arbitration functions to describe agent reaction to deviation. The term is later introduced in \cite{ocfarb}.}
\begin{description}
	\item[{\em The Conservative Arbitration Function}:] under this function, denoted $\cal A_c$, deviators receive nothing from non-deviators. Hence, $\alpha_{\vec c} \equiv 0$ for any deviation. When reasoning about the desirability of a deviation under $\cal A_c$, $S$ has no incentive to retain any of its resources in coalitions with non-deviators: it will not be paid from those coalitions.
	\item[{\em The Sensitive Arbitration Function}:] rather than assuming that agents will refuse outright to cooperate with deviators, one can take a slightly more ``lenient'' approach. Under the sensitive arbitration function, denoted $\cal A_s$, if a coalition $\vec c$ is changed by the deviation of $S$, then all agents in $\supp(\vec c)$ refuse to cooperate with the agents in $S$ in other coalitions. However, if none of the agents in a (possibly different) coalition $\vec c'$ were affected by the deviation, $S$ is still allowed to retain (all of) its original payoffs from $\vec c'$ under $(\CS,\vec x)$.
	\item[{\em The Refined Arbitration Function}:] the refined arbitration function, denoted $\cal A_r$ assumes an even more lenient approach to deviation. Under $\cal A_r$, if a coalition $\vec c$ is unchanged by a deviation of a set $S$, $S$ is allowed to keep (all of) its original payoffs from $\vec c$ under $(\CS,\vec x)$.
	\item[{\em The Optimistic Arbitration Function}:] finally, agents may be highly agreeable to changes in their coalitions. Suppose that the deviation of a set $S$ from a coalition $\vec c \in \CS\setminus\CS|_S$ is $\vec d(\vec c)$; after $S$ deviates, the value of $\vec c$ is reduced to $v(\vec c - \vec d(\vec c))$. Under $\cal A_o$, the payoff to $S$ from $\vec c$ is simply $v(\vec c - \vec d(\vec c)) - \sum_{i \in N\setminus S} x(\vec c)^i$; that is, if $S$ wants to receive payoff from $\vec c$, it must pay the cost of it withdrawing resources from $\vec c$, and ensure that agents in $N\setminus S$ can still retain their original payoffs.  	
\end{description}
The following example highlights the reasoning behind some of the arbitration functions presented above.

\begin{example}
	Suppose that the current coalition structure contains three coalitions, $\vec c_1,\vec c_2$ and $\vec c_{12}$, in addition to several other coalitions. We have that $\supp(\vec c_1) = \{1,3\}, \supp(\vec c_2) = \{2,3\}$ and $\supp(\vec c_{12}) = \{1,2,3\}$. That is, $\vec c_1$ is a coalition in which agents 1 and 3 collaborate, in $\vec c_2$ agents 2 and 3 collaborate, and in $\vec c_{12}$ agents 1, 2 and 3 work together. Now, suppose that agent 3 wishes to withdraw his resources from the coalition $\vec c_1$. Under the conservative arbitration function, agent 3 will not be paid from neither $\vec c_2$ nor $\vec c_{12}$; under the sensitive arbitration function, agent 3 can only expect payoff from $\vec c_2$: since he changed $\vec c_1$ ---a coalition containing an agent in the support of $\vec c_{12}$--- he cannot expect to receive any payoffs from $\vec c_{12}$. Under the refined arbitration function, both $\vec c_2$ and $\vec c_{12}$ will allow agent 3 to retain any payoffs it received from them, as they were unaffected by the deviation.
\end{example}
We observe that the amount that the coalition $\vec c$ needs to pay a deviating set can be determined by the effects that $S$ has on {\em other} coalitions. For example, under the sensitive arbitration function, $\vec c$ may choose not to pay $S$, despite the fact that $\vec c$ was not affected at all by the deviation.

\section{Finding an Optimal Coalition Structure}\label{sec:optval}
We now begin our formal computational analysis of OCF games, starting with the fundamental problem of finding optimal coalition structures.
We assume that the reader is familiar with standard notions of computational complexity and complexity classes (see \cite{gj79np} for an overview).
We begin by describing the formal computational model that we study in this work. 
Given a discrete OCF game $\cal G = \tup{N,\vec W,v}$, we assume that the value $v(\vec c)$ for all $\vec c \in \cal W$ is computable in polynomial time. Moreover, we assume that the arbitration function can be computed in polynomial time for all inputs. In other words, we assume that we have oracle access to both $v$ and $\cal A$. 
This assumption does not trivially hold---a naive representation of the function $v$ is a list of $|\cal W| = \prod_{i = 1}^n (W_i+1)$ values, one for each possible coalition that the agents may form; since $W_i \ge 1$ for all $i \in N$, it follows that $|\cal W| \ge 2^n$, i.e., it is exponential in the natural problem parameters. 

Studying games with a succinct representation, be it OCF or non-OCF games, is an interesting problem in its own right. In our work, we do not discuss representation issues in OCF games, except Section \ref{sec:LBG}, where we focus on the so-called {\em linear bottleneck games}. Instead, rather than focusing on a particular succinct representation, we make the following simplifying assumption: {\em the agents cannot form arbitrarily large coalitions}---i.e., each coalition size is bounded by a constant. 
Formally, we capture this notion in the following definition.
\begin{definition}[$k$-OCF Games]\label{def:k-ocf}
An OCF game $\cal G = \tup{N,\vec W,v}$ is a $k$-OCF game if, for all $\vec c \in \cal W$, if $|\supp(\vec c)| >k$ then $v(\vec c) = 0$.
\end{definition}
Definition~\ref{def:k-ocf} applies to several real-life scenarios where overlapping coalitions form. In many market scenarios, transactions are performed involving only few parties; in social network applications, agents form pairwise coalitions; in many large-scale collaborative projects, small teams are formed to tackle various aspects, as large teams of collaborators tend to be less efficient. In order to simplify notation, given a set of agents $i_1,\dots,i_k$, let us write $v_{i_1,\dots,i_k}(w_{i_1},\dots,w_{i_k})$ to denote the value of $v$ when agent $i_1$ contributes $w_{i_1}$, agent $i_2$ contributes $w_{i_2}$ and so on. For example, $v_{i,j}(w_i,w_j)$ is the value of agents $i$ and $j$ collaborating, where agent $i$ contributes $w_i$ and agent $j$ contributes $w_j$. We note that the number of coalitions that can have a positive value in $k$-OCF games is $\binom{n}{k}\wmax^k$, so if $k$ is a constant, then the characteristic function for a $k$-OCF game can be represented by a polynomial number of bits.

\subsection{Finding an optimal coalition structure}

The problem of ``finding an optimal coalition structure'' in an OCF setting can be rephrased as follows.
Given a $k$-OCF game $\cal G = \tup{N,\vec W,v}$, we are interested in the following question: is the maximal (``optimal'') profit that agents expect to (collectively) accumulate
by forming an overlapping coalition structure 
greater than a given value?
\begin{definition}[$\optval$]\label{def:optval}
An instance of the $\optval$ problem is given by a discrete OCF game $\cal G = \tup{N,\vec W,v}$, a coalition $\vec c \in \cal W$ and a value $V \in \Q$, and is denoted by $\optval(\cal G,\vec c,V)$. It is a ``yes'' instance if and only if $v^*(\vec c) \ge V$.
\end{definition}
Letting $\wmax = \max_{i \in N}W_i$, the following simple proposition is an important first step in computing solutions for OCF games.
\begin{proposition}\label{prop:optval}
Given an OCF game $\cal G$, a coalition $\vec c$, and a parameter $V$,
if $|\supp(\vec c)| = m$, then $\optval(\cal G,\vec c,V)$ is decidable in time polynomial in $(\wmax+1)^m$.
\end{proposition}
\begin{proof}
We observe that
$$v^*(\vec c) = \max\left\{v(\vec c),\{v^*(\vec c - \vec d) + v(\vec d) \mid \vec d \le \vec c;\vec d \ne \vec c\}\right\}.$$
Observe that the number of coalitions $\vec d$ such that $\vec d \le \vec c$ is at most $(\wmax+1)^m$, so computing $v^*(\vec c)$ requires computing at most $(\wmax+1)^m$ values of $v^*$, which require a comparison of at most $(\wmax+1)^m$ values each, for a total running time that is polynomial in $(\wmax+1)^m$.
\end{proof}
Proposition~\ref{prop:optval} implies that if $|\supp(\vec c)|$ is a constant, and $\wmax$ is polynomial in $n$, then $\optval(\cal G,\vec c,V)$ can be computed in time polynomial in $n$. This result is the best one can hope for: \cite{ocfgeb} show via a reduction from the \textsc{Knapsack} problem~\citep{gj79np}, that when agent weights are large, $\optval$ is NP-complete. We now turn to study $\optval$ when the group of agents involved is large---i.e., $|\supp(\vec c)|$ is not a constant. We stress that when $|\supp(\vec c)|$ is not a constant, this simply means that a large group of agents wants to form a coalition structure. The coalition sizes in the coalition structure will still be small in our setting. Proposition~\ref{prop:optval} was concerned with a small number of agents wanting to find the best way to share resources among themselves; when the number of agents is not small, the actual coalitions that they form may still be. For example, if $|\supp(\vec c)| = 100$, this means that there are 100 agents whose available resources are given by $\vec c$. Suppose that the agents may only conduct pairwise interactions; in that case, agents simply wish to find the best way for them to form smaller coalitions using the resources available as per $\vec c$. 

We begin by showing the following negative result.
\gc{It is kind of confusing to say that
``the group of agents involved is large'', but then talk about 2-OCF games, and, immediately after the proof, say that ``agent interactions are limited to pairwise interactions''.
We could say more clearly / state again what ``coalition size'' is, as opposed to ``agent interactions'' in this setting.
However, all this might be heralding good news, in a sense - we could give this hardness result a positive spin:
In light of proposition 2, one could indeed argue that, in terms of having some hope to answer optval, we are limited to ``small $k$''-OCF games, say $2$-OCF games, with
{\em a limited support of, say, $|\supp(\vec c)=2$}. That's actually ok: In the numerous real-world settings where 2-OCF games are present, {\em e.g., matching markets}, we do not actually
care about coalition sizes (i.e., support) of $>2$: since no ``third party'' can ever hope to get anything from ``getting involved with'' a couple(since their collective value is then 0), 
why even bother contributing resources to a coalition with these agents? Thus, we could argue the following: in real-life scenarios where ``pairs'' of agents have already begun interacting with each other before
others arrive, and where newcomers do not seem to be adding value but -on the contrary- seem to be ``messing up'' the order of things, the hardness problem will not even present itself.
Does this sound about right?}
\yz{I have made some slight rewording of the explanation, tell me if the above is reasonable.\\ Vangelis: I am fine with the current wording. In fact we are too detailed and we could even cut this discussion shorter, it should be clear what we mean here.}

\begin{proposition}\label{prop:hard}
$\optval$ is NP-complete even if one assumes that $\cal G$ is a 2-OCF game, and that $\wmax = 3$.
\end{proposition}
\begin{proof}
First, we observe that $\optval$ is in NP; it suffices to guess
a coalition structure $\CS$ such that $\vec w(\CS) = \vec c$ and check whether $v(\CS)\ge V$; note that
the size of this coalition structure is at most $n\wmax$,
which is polynomial in the input size, assuming that $\wmax$ is polynomial in $n$.

For the hardness proof, we provide a reduction from
{\sc Exact Cover by 3-Sets (X3C)},~\cite{gj79np}.
Recall that an instance $\cal X = \tup{A,\cal S,\ell}$ of {\sc X3C} is given by
a finite set $A$, $|A|=3\ell$,
and a collection of subsets $\cal S  = \{S_1,\dots,S_t\}$
such that $S_j \subseteq A$ and $|S_j| = 3$ for all $j=1, \dots, t$. It is a ``yes''-instance
if $A$ can be exactly covered by sets from $\cal S$; that is, if there exists a subset $\cal S'\subseteq \cal S$ such that $\bigcup_{S \in \cal S'} S = A$, and for any two $S,T \in \cal S'$ we have $S \cap T = \emptyset$.

Given an instance $\cal X = \tup{A,\cal S}$ of {\sc X3C},
we construct a discrete 2-OCF
game $\cal G(\cal X) = \tup{N,v}$ with $\wmax=3$, as follows. We have an agent $a_i$ of weight 1
for every element $i \in A$ and an agent $a_S$ with weight 3
for every $S \in \cal S$. The characteristic function is defined as follows:
if $i\in S$, then the value of both $a_i$ and $a_S$ forming a coalition where each of them contributes a weight of 1 is 2; that is, if $\supp(\vec c) = \{a_i,a_S\}$ and 
agent contributions are $c^{a_i} = c^{a_S} = 1$, then $v(\vec c) = 2$. Moreover, agents corresponding to sets in $\cal S$ can generate profits by working alone: if they dedicate all their weight to forming a singleton coalition, they generate a profit of 5. In other words, if $\supp(\vec c) = \{a_S\}$ and $c^{a_S} = 3$ then $v(\vec c) = 5$.
The value of every other partial coalition in the game $\cal G(\cal X)$ is $0$.

Let $S\in \cal S$ be $\{x, y, z\}$, and consider the set of agents
$G_S=\{a_S, a_x, a_y, a_z\}$.
Collectively, the agents
in $G_S$ can earn $6$ if $a_S$ forms a partial coalition
with each of $a_x$, $a_y$, and $a_z$, and contributes
one unit of weight to each of these coalitions;
in any other coalition structure, $G_S$ earns at most $5$.
Hence, $\cal X = \tup{A,\cal S}$ admits an exact cover if and only if
$v^*(\vec W)\ge 6\frac{|A|}{3}+5(t-\frac{|A|}{3})=5t+\frac{|A|}{3}$. Let $V$ be this value.

To see the correctness of our reduction, suppose that $\cal X$ is a ``yes'' instance of {\sc X3C}. Then there exists some $\cal S' \subseteq \cal S$ of size $\frac{|A|}{3}$ that exactly covers $A$. In that case,
$$v^*(\vec W) \ge \sum_{S \in \cal S'}v^*(\vec W^{G_S}) + \sum_{S \notin \cal S'}v^*(\vec W^{G_S}) = 6\frac{|A|}{3} + 5(t-\frac{|A|}{3}) = V.$$
On the other hand, suppose that $v^*(\vec W) \ge 6\frac{|A|}{3} + 5(t - \frac{|A|}{3})$. Then this means that there is a subset of $\cal S$, $\cal S'$, of size at least $\frac{|A|}{3}$ such that for all $S \in \cal S'$ we have that $v^*(\vec W^{G_S}) = 6$. Note that since $|S| = 3$ for all $S \in \cal S'$ and since $S\cap T = \emptyset$ for all $S, T \in \cal S'$, it must be that $|\cal S'| = \frac{|A|}{3}$, i.e., $\cal S'$ is a partition of $A$; thus, $\cal S'$ is indeed an exact cover of $A$.
\end{proof}
Proposition~\ref{prop:hard} severely limits our prospects of computing optimal coalition structures: even if agents are limited to pairwise interactions, and agent weights are small constants, $\optval$ remains hard. Observe that the hardness of $\optval$ implies the hardness of most other problems of interest in our work: the most that a set can gain by deviating, the stability of a given game, and the possibility of deviation from a given outcome of a game. Thus, in order to proceed, we must first identify some limiting conditions that make $\optval$ computationally tractable.

\subsection{Limiting Interactions in OCF Games}
\cite{Demange2004} shows that if one assumes a hierarchical agent interaction structure in a cooperative game, then the core of the game is not empty; moreover, it is possible to find a core imputation in polynomial time. We now show how to adapt this interesting result to the OCF setting. The formal model we propose is not unlike the one presented by Demange. An agent {\em interaction graph} is a graph $\Gamma =\tup{N,E}$, where the edges in $E$ represent valid agent interactions; given an OCF game $\cal G = \tup{N,\vec W,v}$, the game $\cal G$ {\em reduced to $\Gamma$}, denoted $\cal G|_\Gamma= \tup{N,v|_\Gamma}$ is defined as follows: for every $\vec c \in \cal W$, if the nodes in $\supp(\vec c)$ induce a connected subgraph in $\Gamma$ then $v|_\Gamma(\vec c) = v(\vec c)$, otherwise $v|_\Gamma(\vec c) = 0$. Such graph restrictions on agent interaction are known in classic cooperative game theory as {\em Myerson graphs} \citep{myerson1977graphs}. Unfortunately, limiting agent interaction is not a sufficient condition for the tractability of $\optval$. As shown in~\cite{ocfgeb}, finding an optimal coalition structure is hard, even for a single agent, if his weight is sufficiently large. Moreover, even if weights are small but one does not limit the size of allowable coalitions, assuming that agents interact based on a hierarchical tree structure does not aid computational complexity, as shown by the following proposition.
\begin{proposition}\label{prop:hard2}
$\optval$ is NP-hard, even for the family of instances $(\cal G|_T,\vec c,V)$ where $T$ is a tree and agent weights are constant.
\end{proposition}
\begin{proof}
Our reduction is from the $\indset$ problem \citep{gj79np}. An instance of $\indset$ is a tuple $\tup{\Gamma,m}$, where $\Gamma = \tup{N,E}$ is a graph and $m$ is an integer. It is a ``yes'' instance if $\Gamma$ contains an independent set of size at least $m$, and is a ``no'' instance otherwise. Recall that an independent set in $\Gamma$ is a subset of vertices $S \subseteq N$ such that if $i,j \in S$, then the edge $\{i,j\}$ is not in $E$. In other words, it is a subset of vertices that do not share an edge with one another. Given an instance $\tup{\Gamma= \tup{N,E},m}$ of $\indset$, we construct the following instance of $\optval$: we set the player set to be $N' = N\cup\{n+1\}$, and have the interaction graph $T$ connect all vertices in $N$ with $n+1$ via an edge. We set all agents in $N$ to have a weight of 1, whereas agent $n+1$ has a weight of 2. The characteristic function $v$ is defined as follows: given a set $S \subseteq N$, if it is an independent set in $\Gamma$, and player $n+1$ allocates a weight of 1 to working with $S$, then the value of the resulting coalition is 1. If $S$ forms a vertex cover of $\Gamma$ (i.e., all vertices in $\Gamma$ share an edge with the vertices in $S$), then the value of the resulting coalition is $\frac{\eps}{|S|+1}$, where $\eps$ is some constant much smaller than 1. In order to form an optimal coalition structure, $n+1$ must allocate a weight of 1 to the maximal independent set, and a weight of 1 to working with the rest of the vertices. This is because if $S$ is an independent set, then $N \setminus S$ is a vertex cover. We conclude that the value of the optimal coalition structure in the resulting game is more than $1+ \frac{\eps}{n-m+1}$ if and only if there exists an independent set of size at least $m$.
\end{proof}
We note that Proposition~\ref{prop:hard2} does not hold for classic cooperative games: this would be a contradiction to Demange's theorem. However, it is worth noting that the game described in Proposition~\ref{prop:hard2} is ``almost'' a classic cooperative game: there is only one agent whose weight is not 1, and his weight is only 2. Indeed, changing the weight of player $n+1$ from 2 to 1 in the above theorem would result in a game with a trivially optimal coalition structure: forming the grand coalition $N\cup\{n+1\}$.

Propositions~\ref{prop:hard} and~\ref{prop:hard2} suggest the following three conditions are {\em necessary} in order to find an optimal coalition structure in polynomial time:
\begin{enumerate}
\item
Limiting agent weights.
\item
Limiting agent interactions.
\item
Limiting coalition size.
\end{enumerate}
Dropping any of these three limitations results in hard instances of $\optval$.
We can, however, compute an optimal structure for games whose interaction graphs are trees in time polynomial in $\wmax$, assuming that $\cal G$ is a 2-OCF game. Given a tree $T = \tup{N,E}$ whose root is some $r \in N$, and a player $i \in N$, let $T_i$ be the subtree rooted at $i$, and $C_i(T)$ be the children of $i$ in $T$. Given a graph $\Gamma$, we refer to the nodes of $\Gamma$ as $N(\Gamma)$. Finally, given a coalition $\vec c \in \cal W$, we write $(\vec c_{-i},w)$ to denote the coalition $\vec c$ with the $i$-th coordinate replaced by $w$. For ease of exposition, we say that $\cal G$ has a {\em tree interaction structure} if there exists some tree $T = \tup{N,E}$ such that $\cal G = \cal G|_T$; this way, $T$ is part of the game description, and does not have to be an additional input.
\begin{theorem}\label{thm:vtree}
For the family of discrete 2-OCF games with a tree interaction structure, $\optval$ can be decided in time polynomial in $\wmax$ and $n$.
\end{theorem}
\begin{proof}
Let $\cal G = \tup{N,\vec W,v}$ be a discrete 2-OCF game. We will show how to compute $v^*(\vec W)$; however, our result easily holds for general $\vec c \in \cal W$: simply look at the game where agent weights are as per $\vec c$. 

We arbitrarily choose some player $r \in N$ to be the root of the interaction graph, and process the players in $N$ from the leaves up to the root. The key observation to make is that in order to find an optimal allocation of the players' resources, each node in the tree needs to decide how much weight is to be allocated to its subtree, how much is to be allocated to collaborating with its parent, and how much is to be allocated to working alone. Note that according to Proposition~\ref{prop:optval}, both $v_i^*(x)$ and $v_{i,j}^*(x,y)$ can be computed in time polynomial in $\wmax$. Given a node $i \in N$, let $C_i$ be the children of $i$ in the interaction graph. Let $v_{T_i}^*(w)$ be the most that the subtree $T_i$ can make if agent $i$ allocates $w$ to working with $T_i$. We note that
$$v_{T_i}^*(w) = \max_{\stackrel{y_i + \sum_{c \in C_i} w_c = w}{\forall c \in C_i:0\le x_c\le W_c}}\{v_i^*(y_i) + \sum_{c\in C_i}v_{i,c}^*(w_c,x_c) + v_{T_c}^*(W_c - x_c)\}.$$
Using this recurrence relation we obtain the following dynamic programming method of finding an optimal coalition structure.

Suppose that we have already computed $v_{T_c}^*(w)$ for all $c \in C_i$ and all $w = 0,\dots,w_c$. Now, let us write $C_i = \{c_1,\dots,c_m\}$, and let $T_{i,j}$ be the subtree $T_i$, but with the subtrees $T_{c_{j+1}},\dots,T_{c_m}$ removed. Let $v_{T_{i,j}}^*(w)$ be the maximal revenue that can be generated if agent $i$ invests $w$ in working with $T_{i,j}$. We also write $T_{i,0}$ to be the tree $T_i$ with all subtrees removed, i.e., the tree comprised of the singleton $\{i\}$. Now, suppose that we have already computed $v_{T_{i,j'}}^*(w)$ for all $j' = 0,\dots,j-1$ and all $w = 0,\dots,W_i$. In that case, we can compute $v_{T_{i,j}}^*(w)$ in time polynomial in $\wmax$ and $n$, as
$$v_{T_{i,j}}^*(w) = \max_{\stackrel{x+z = w}{y \le W_{c_j}}}\{v_{i,c_j}^*(x,y) + v_{T_{i,j-1}}^*(z)+v_{T_{c_j}}^*(W_{c_j}-y)\},$$
allowing us to find the value of $v_{T_{i,j}}^*(w)$ in time polynomial in $\wmax$, and $v_{T_i}^*(0)$, $\dots$, $v_{T_i}^*(W_i)$ in time polynomial in $\wmax$ but linear in $|C_i|$. Going through all the nodes from the leaves to the root, we obtain that the value of the optimal coalition structure is simply the value $v_{T_r}^*(W_r)$.
\end{proof}

\section{Computing Optimal Deviations}\label{sec:arbval}
In Section~\ref{sec:optval} we identified three key conditions for the computational tractability of {\em finding an optimal coalition structure} in OCF games. Before we turn to computational aspects of {\em stability} in OCF games, let us study the problem of deciding the most that a set can obtain by deviating from a given outcome. Formally, we are interested in the following problem:
\begin{definition}[$\arbval$]\label{def:arbval}
An instance of the $\arbval$ problem is given by a tuple of the form $\tup{\cal G,\cal A,S,(\CS,\vec x),V}$, where $\cal G = \tup{N,v}$ is an OCF game, $\cal A$ is some arbitration function, $S$ is a subset of $N$, $(\CS,\vec x)$ is an outcome of $\cal G$, and $V\in\mathbb Q$ is some parameter. It is a ``yes'' instance if $\cal A^*(\CS,\vec x,S) \ge V$, and is a ``no'' instance otherwise.
\end{definition}
We note that any hardness results obtained for $\optval$ are immediately inherited by $\arbval$, as the two problems are identical if one assumes that non-deviators behave according to the conservative arbitration function---i.e., give zero payoffs to deviators regardless of the nature of their deviation.
However, $\arbval$ is a considerably more complex problem than $\optval$; in order to ensure that $\arbval$ can be decided in polynomial time, one must make assumptions not only on the structure of the game $\cal G$, but also on the way that agents react to deviation. In a sense, computational complexity can arise by making the payment structure of $\cal A$ sufficiently complex. This is shown in the following proposition.

\begin{proposition}\label{prop:arbnp}
If there exists a polynomial time algorithm that
can decide $\arbval$, when restricted to instances $\tup{\cal G,\cal A,S,(\CS,\vec x),V}$ such that $\cal G = \tup{N,\vec W,v}$ is a $2$ player game, then $P = \NP$.
\end{proposition}
\begin{proof}
We will show that
if such an algorithm exists, it can be used to solve instances
of $\setcover$~\cite{gj79np}. Recall that an instance
of $\setcover$ is given by a set of elements $A$,
a collection of subsets $\cal S = \{S_1,\dots,S_t\} \subseteq 2^A$
and $\ell \in \N$; it is a ``yes''-instance if $A$ can be covered
by at most $\ell$ sets from $\cal S$.

Given an instance of $\setcover$, $\tup{A,\cal S,\ell}$
such that $|\cal S| = t$,
consider a $2$-player discrete OCF game where $w_1=w_2=t+2$.
We define $v$ in the following manner. First, players get a payoff of 1 for each unit of resource devoted to working alone, i.e. $v_1(x) = v_2(x) = x$ for all $x \in \{0,\dots,t+2\}$. We also set $v_{1,2}(1,1) = 2$, and $v_{1,2}(2,2) = 10(t+2)$. All other coalitions can have an arbitrary value (we assume with no loss of generality that it is 0).

We define an outcome $(\CS,\vec x)$ as follows. players 1 and 2 form $t$ coalitions where each one of them devotes 1 unit of resources to working together, and an additional coalition where both invest 2 units; that is, $\CS = (\vec c_1,\dots,\vec c_t,\vec d)$, such that $\vec c_j = (1,1)$ for all $1\le j\le t$, and $\vec d = (2,2)$. We define $\vec x = (\vec x_1,\dots,\vec x_t,\vec y)$ as follows: $\vec x_j = (0,2)$ is the payoff division from coalition $\vec c_j$, and $\vec y = (5(t+2),5(t+2))$. In other words, we allocate the payoffs from $\vec c_1,\dots,\vec c_t$ to player 2, and split the payoff from $\vec d$ equally.

We define the arbitration function $\cal A$ as follows. Given that player 1 wishes to deviate from $(\CS,\vec x)$, by withdrawing only from the coalitions $\vec c_{j_1},\dots,\vec c_{j_s}$, he receives no payoff from the coalitions $\vec c_1,\dots,\vec c_t$, and will only get to keep his payoff from $\vec d$ if the collection $\{S_j \in \cal S\mid j \notin\{j_1,\dots,j_s\}\}$ is a set cover of $A$. For any other input, we can define an arbitrary output for $\cal A$ (without loss of generality, let us assume that $\cal A$ behaves as the refined arbitration function on other inputs). Under this arbitration function, player 1 wants to withdraw as much resources as possible from $\vec c_1,\dots,\vec c_t$, but do so in a manner that the coalitions he keeps intact correspond to a set cover of $A$.

We observe that $\cal A^*(\CS,\vec x,\{1\}) \ge 5(t+2) + t - \ell$ if and only if $\tup{A,\cal S,\ell}$ is a ``yes'' instance of $\setcover$. First, if there is a subset $\cal S' \subseteq \cal S$ such that $|\cal S'|\le \ell$, then by withdrawing from the coalitions corresponding to $\cal S \setminus \cal S'$, and allocating the withdrawn resources to working alone, player 1 ensures that he receives a payoff of at least $5(t+2) + t - \ell$; on the other hand, if there exists a set of coalitions $\CS' = (\vec c_{j_1},\dots,\vec c_{j_s})$ such that withdrawing from $\CS'$ ensures that the payoff to player 1 is at least $5(t+2) + t - \ell$, then it must be that the set of coalitions that player 1 chose to withdraw from does not contain $\vec d$, and corresponds to a set $\cal S' \subseteq \cal S$ such that $\cal S \setminus \cal S'$ is a set cover of $A$.
\end{proof}

\begin{remark}\label{rem:optval_vs_stabarb}
We contrast Proposition~\ref{prop:arbnp} with Proposition~\ref{prop:optval}: computing the most that a set can gain with a given set of resources is computationally much easier than deciding what is the most it stands to gain by deviating. This issue does not arise in classic cooperative games: a set assesses the desirability of deviation by considering the most it can make on its own, which only requires computing $v(S)$.
\end{remark}

It seems that the hardness of deciding $\arbval$ stems from the fact that the payoff from a coalition $\vec c$ to a deviating set $S$ is determined by the way $S$ affects other coalitions. In the reduction used in Proposition~\ref{prop:arbnp}, the payoff to player 1 from the coalition $\vec d$ was determined by the deviation from other coalitions. In other words, the arbitration function determines the payoff to a deviating set $S$ based on the global behavior of $S$.

This observation motivates the following definition.

\begin{definition}
An arbitration function $\cal A$ is {\em local} if the payoff from a coalition $\vec c$ depends only on the effect of the deviating set $S$ on $\vec c$, regardless of the input to $\cal A$. In other words, for any game $\cal G = \tup{N,\vec W,v}$, any outcome $(\CS,\vec x)$, any set $S \subseteq N$ and any deviation $\CS'$ of $S$ from $\CS$, the payoff to $S$ from $\vec c \in \CS\setminus\CS|_S$ depends only on $\vec d(\vec c)$, $\vec x(\vec c)$, $S$, and $\cal G$.
\end{definition}
We note that the conservative, refined and optimistic arbitration functions are local: the payoff from the the conservative arbitration function is 0 for all inputs; the payment from the refined arbitration function is $x(S)$ if $\vec d(\vec c) = 0^n$ and is 0 otherwise, and the payment from the optimistic arbitration function is $\max\{v(\vec c - \vec d(\vec c)) - \sum_{i \notin S}x(\vec c)^i,0\}$.
In contrast, the arbitration function used in the proof of Theorem~\ref{prop:arbnp}
is non-local. Another example of a non-local
arbitration function is the sensitive arbitration function, as the payoff to a set from a coalition $\vec c$ depends on which agents in the support of $\vec c$ were hurt by the deviation of $S$ from other coalitions.

When one is limited to the class of local arbitration functions, it is indeed possible to decide $\arbval$ in time polynomial in $|\CS|$ and $\wmax^{|S|}$, where $S$ is the deviating set. 
\begin{theorem}\label{thm:localarb}
$\arbval$ is decidable in time polynomial in $|\CS|$ and $\wmax^{|S|}$ for all instances $\tup{\cal G,\cal A,S,(\CS,\vec x),V}$ such that $\cal A$ is local.
\end{theorem}
\begin{proof}
We first observe that a coalition structure $\CS$ has at most
$(\wmax+1)|S|$ coalitions that involve players in $S$.
Given a coalition structure $\CS$, let $\CS'$ be the set of coalitions 
that are supported by both $S$ and $N\setminus S$; i.e., $\CS' = \{\vec c \in \CS\mid \supp(\vec c) \cap S \ne \emptyset;\supp(\vec c) \cap N\setminus S \ne \emptyset\}$. For every $\vec c \in \CS'$, recall that $\vec c^S$ is the coalition $\vec c$ with the coordinates corresponding to members of $N \setminus S$ equal to 0; that is, $\vec c^S$ is the contributions of $S$ to the coalition $\vec c$.
Now, suppose that players in $S$ invest
$\vec s\in \cal W(S)$ units of resources in partial coalitions among themselves ---i.e., $\vec w(\CS|_S) = \vec s$---
and want to withdraw an additional $\vec t\in \cal W(S)$ from $\CS$. They would get
$v^*(\vec s+\vec t)$ from working on their own, plus the most that $S$ can
get from the arbitration function, which depends on the coalitions
affected by this deviation. Thus, in order to determine the most that $S$ can get by deviating from $(\CS,\vec x)$, given that it must withdraw a total of $\vec t$ resources, we must determine how to best withdraw those $\vec t$ resources from $\CS'$. We write $\CS' = (\vec c_1,\dots,\vec c_m)$. 

Let us denote by $A(\vec t;\ell)$ the most that the
arbitration function $\cal A$ will give $S$ if they withdraw $\vec t$ resources from the first $\ell$ coalitions, where $1 \le \ell \le m$. We also write $\alpha_{\vec c}(\vec t)$ to be the payoff to $S$ from the coalition $\vec c$ if $S$ withdraws $\vec t$ from $\vec c$. Since agents may not withdraw more resources than they have invested in a coalition, we set $\alpha_{\vec c}(\vec t) = -\infty$, if $\vec t$ is greater than $\vec c$ in any coordinate. We note that $\alpha_{\vec c}(0^n)$ is the amount that $S$ receives from $\vec c$ if it does not withdraw any resources from $\vec c$. The value $\alpha_{\vec c}(0^n)$ can be any value in $\R$, but the key observation is that if the coalition $\vec c$ has no resources withdrawn from it, it pays $\alpha_{\vec c}(0^n)$ to $S$, regardless of how $S$ affects other coalitions.

By definition, $A(\vec t;1) = \alpha_{\vec c_1}(\vec t)$, for other coalitions we have
$$
A(\vec t;\ell) =
\max\{A(\vec y;\ell-1)+ \alpha_{\vec c_\ell}(\vec t - \vec y)\mid 0^n\le \vec y\le \vec t\}.
$$
This shows that we can compute $A(\vec t;m)$ in
$\cal O(m(\wmax+1)^{|S|})$ steps.
Finally, $\cal A^*(\CS,\vec x,S)$ can be computed as
$\max\{v^*(\vec s+\vec t)+A(\vec t;m)\mid 0^n \le \vec t \le \vec W^S-\vec s\}$, which concludes the proof.
\end{proof}

Coming back now to imposing interaction structures on the cooperation among agents, we provide an analogous positive result to that in Section \ref{sec:optval}. 
In particular, as is the case for computing an optimal coalition structure for games where the interaction graph is a tree, computing the most that a set can get by deviating can be done in time polynomial in $n$ and $\wmax$ if the arbitration function is local and the interaction graph is a tree. In fact, one can see Theorem~\ref{thm:arbtree} as an immediate corollary of Theorem~\ref{thm:vtree}.
\begin{theorem}\label{thm:arbtree}
$\arbval$ is decidable in time polynomial in $n$ and $\wmax$ for all instances $\tup{\cal G,\cal A,S(\CS,\vec x),V}$ such that $\cal A$ is local and such that $\cal G$ is a discrete 2-OCF game with a tree interaction structure. 
\end{theorem}
\begin{proof}
Again, we choose an arbitrary $r \in N$ to be the root of the interaction tree. 
Consider some player $i\in S$; $i$ needs to decide how much weight to allocate 
to each non-$S$ neighbor, given the amount that he gives to his parent. 
Let $D_i$ be the set of non-$S$ neighbors of $i$, that $i$ gets some payoff from 
interacting with, under $(\CS,\vec x)$. Let us denote by $\alpha_i(w)$ 
the most that $i$ can make if he devotes $w$ to interacting with $D_i$. 
In order to compute $\alpha_i(w)$, we use dynamic programming. 
We fix an ordering of $D_i$ and 
let $\alpha_i(w;j)$ be the most that player $i$ can get by keeping a weight of $w$ 
in the interaction with the first $j$ agents in $D_i$. 
We also denote by $\cal A_i^*(w;j)$ the payoff to player $i$ if he keeps a weight of $w$ 
in the interaction with the $j$-th player in $D_i$. 
By Theorem~\ref{thm:localarb}, $\cal A_i^*(w;j)$ can be computed in time polynomial in $\wmax$. Further, 
we have $\alpha_i(w;0) = 0$ for all $w$ and
$$
\alpha_i(w;j) = \max\{\alpha_i(x;j-1) + \cal A_i^*(w-x;j)\mid 0 \le x \le w\},
$$
Set 
$$
\bar{v}_i^*(w) = \max\{v_i^*(x)+\alpha_i(w-x) \mid 0 \le x \le w\}.
$$
We now replace player $i$'s $v_i^*$ by $\bar{v}_i^*$ for the purpose 
of computing $v^*(S)$, which is 
doable in polynomial time according to Theorem~\ref{thm:vtree}; this will give us the most that $S$ can 
get by deviating under $\cal A$.
\end{proof}
We observe that Theorem~\ref{thm:arbtree}
holds even if the overall interaction graph is not a tree; it suffices that 
the deviating set $S$ is an acyclic subgraph of the interaction graph. 
This is because in 2-OCF games, interactions can only be between pairs of players; thus, if an agent $i \in S$ decides to withdraw resources from a coalition 
$\vec c$, that coalition can contain at most one other agent $j$. Therefore, no other agent in $S$ is 
concerned with $\vec c$, i.e., agents in $S$ need to keep individual track of how much they withdraw from coalitions supported by $N \setminus S$, which significantly simplifies the computational process of deciding where to withdraw resources from.

\section{Computing $\cal A$-Stable Outcomes}\label{sec:stable}
We now turn to a different computational problem. Having provided efficient procedures for computing optimal deviations and coalition structures in discrete 2-OCF games with tree interaction structures, we are ready to analyze the computational complexity of stability in this class of games. Recall that an OCF game is $\cal A$-stable if there exists some outcome $(\CS,\vec x)$ such that no subset of $N$ can profitably deviate from $(\CS,\vec x)$, i.e., if for all $S \subseteq N$ we have that $$p_S(\CS,\vec x) \ge \cal A^*(\CS,\vec x,S).$$
In this context, we are first interested in the following problem.
\begin{definition}[$\checkcore$]\label{def:checkcore}
An instance of $\checkcore$ is given by a discrete OCF game $\cal G = \tup{N,\vec W,v}$, an arbitration function $\cal A$, and an outcome $(\CS,\vec x)$; it is a ``yes'' instance if $(\CS,\vec x)$ is in the $\cal A$-core of $\cal G$, and a ``no'' instance otherwise. 
\end{definition}

We are now ready to present an algorithm
for checking whether a given outcome is in the $\cal A$-core.
This problem is closely related to that of computing $\cal A^*$:
an outcome $(\CS,\vec x)$ is in the $\cal A$-core if and only if 
the {\em excess} $e(\CS, \vec x, S) = \cal A^*(\CS,\vec x,S) - p_S(\CS, \vec x)$
is non-positive for all coalitions $S\subseteq N$. Thus, we need to check
whether there exists
a subset $S \subseteq N$ with $e(\CS, \vec x, S) > 0$.
Note that it suffices to limit attention
to subsets of $N$ that form connected subgraphs of the interaction graph of the game: if $e(\CS, \vec x, S) > 0$
and $S$ is not connected, then some connected component $S'$ of $S$
also satisfies $e(\CS, \vec x, S') > 0$.

\begin{theorem}\label{thm:checkcore}
If $\cal G$ is a 2-OCF game with a tree interaction structure, and $\cal A$ is local, then $\checkcore$ is decidable in time polynomial in $n$ and $\wmax$ for any instance $\tup{\cal G,\cal A,(\CS,\vec x)}$.
\end{theorem}
\begin{proof}
Fixing an outcome $(\CS, \vec x)$, let us write $p_i:=p_i(\CS, \vec x)$ 
for all $i\in N$. 
As in our previous results, we pick an arbitrary $r\in N$ as a root. 
We say that $S \subseteq N$ is {\em rooted}
at $i \in N$ if $i \in S$ and the members of $S$ 
form a subtree of $T_i$ - the tree rooted in $i$. We observe that every 
set $S \subseteq N$ is rooted at a unique $i\in N$.
Given a vertex $i$, let $E_i$ denote the maximum excess of a set rooted at $i$, that is:
$$E_i = \max\left\{e(\CS,\vec x,S) \mid S \mbox{ is rooted in } T_i\right\}.$$
Clearly, $(\CS,\vec x)$ is not $\cal A$-stable
if and only if $E_i > 0$ for some $i\in N$. It remains to show that all $E_i$ can be computed in time polynomial in $n$ and $\wmax$.
As before, we proceed from the leaves to the root, and terminate (and report that
$(\CS, \vec x)$ is not $\cal A$-stable) if we discover a vertex $i$
with $E_i>0$. If $E_i\le 0$ for all $i\in N$, we report that
$(\CS, \vec x)$ is $\cal A$-stable.

Given two agents $i, j\in N$, let $w_{i,j}$ denote the total weight that $i$ assigns
to interacting with $j$ under $\CS$, i.e.,
$$w_{i,j} = \sum_{\vec c:\supp(\vec c) = \{i,j\}} c^i.$$ 
Observe that since $\cal G$ is a 2-OCF game, it is no loss of generality to assume that $i$ only invests weight in interacting with $j$ only in coalitions that are supported by $i$ and $j$ alone. We begin by defining two auxiliary values.
First, given a neighbor $j$ of $i$, 
we define $\alpha_{i,j}(w)$ to be the most that $\cal A$ will 
give $i$ if he keeps a total weight of $w\le w_{i,j}$ 
in the coalitions that he formed with $j$ in $(\CS, \vec x)$; 
by Theorem~\ref{thm:localarb}, 

The value $\alpha_{i,j}(w)$ is computable in time polynomial in $\wmax$. 
Second, we define $D_i(w)$ to be the maximum excess of a subset rooted at $i$ 
if $i$ were to contribute $w$ to $T_i$ and nothing to his parent $p(i)$. 
In this notation, 
$$
E_i = \max\{D_i(w)+\alpha_{i,p(i)}(y)\mid w+y = w_i, w\ge w_i - w_{i, p(i)}\};
$$
the condition $w \ge w_i - w_{i,p(i)}$ ensures that $p(i)$ is not among the deviators.
It remains to show how to compute $D_i(w)$ in time $\poly(n,\wmax)$ 
for all $i\in N$ and $w$ such that $w_i-w_{i, p(i)}\le w \le w_i$.

Consider an agent $i$ with children
$C_i=\{i_1, \dots, i_\ell\}$, and suppose that we have computed $D_{i_j}(z)$
for each $i_j\in C_i$ and each $z$ such that  $w_{i_j}-w_{i_j, i}\le z\le  w_{i_j}$
(this encompasses the possibility that $i$ is a leaf, as $C_i=\emptyset$ in that case).
For $j=0, \dots, \ell$, let $T_{i, j}$ be the tree
obtained from $T_i$ by removing subtrees rooted at $i_{j+1}, \dots, i_\ell$.
Let $D_i(w;j)$ be the maximum excess of a set rooted at $i$
that is fully contained in $T_{i, j}$, assuming that $i$ contributes $w$ 
to $T_{i, j}$ and nothing to his parent or his children $i_{j+1}, \dots, i_\ell$;
we have $D_i(w) = D_i(w;\ell)$.
We will compute $D_i(w;j)$ by induction on $j$.

We have $D_i(w; 0) = v_i^*(w) - p_i$ for all 
$w = w_i-w_{i, p(i)}, \dots, w_i$.
Now, consider $j>0$.
Agent $i$ can either include $i_j$ in the deviating set or 
deviate (partially or fully) from the coalitions that 
it forms with $i_j$ in $(\CS, \vec x)$. 
Thus, $D_i(w; j)=\max\{D_1, D_2\}$, where
$$
D_1 = \max_{\stackrel{y=0, \dots, w}{z = 0, \dots, w_{i_j}}}
\{D_i(y; j-1) + v_{i,i_j}^*(w-y, z) + D_j(w_{i_j}-z)\}.
$$
and
$$
D_2 = \max_{z=0, \dots, w_{i, i_j}}
\{D_i(w-z; j-1) + \alpha_{i,i_j}(z)\}.
$$
Since both quantities $D_1$ and $D_2$ can be computed in time polynomial in $\wmax$, 
we can efficiently compute $D_i(w;j)$, and hence also 
$D_i(w)$ and $E_i$.
\end{proof}
Not only can the algorithm described in Theorem~\ref{thm:checkcore} decide $\checkcore$ in time polynomial in $n$ and $\wmax$, it can also be used to decide the following, closely related problem. Instead of asking if a given outcome $(\CS,\vec x)$ is $\cal A$-stable, we are also interested in deciding whether a given coalition structure $\CS$ can be stabilized, i.e., whether there exists some division of payoffs among agents in such a way that the resulting outcome is $\cal A$-stable. This is formalized in the following definition.
\begin{definition}\label{def:isstable}
An instance of $\isstable$ is given by a discrete OCF game $\cal G = \tup{N,\vec W,v}$, an arbitration function $\cal A$ and a coalition structure $\CS$. 
\end{definition}
Given a coalition structure $CS = (\vec c_1,\dots,\vec c_m)$, an instance of $\isstable$ $\tup{\cal G,\cal A,\CS}$ is a ``yes'' instance if and only if there exists a payoff division $\vec x\in I(\CS)$, $\vec x = (\vec x_1,\dots,\vec x_m)$, that satisfies the following system of constraints.
\begin{eqnarray}\label{eq:isstable}
\sum_{i \in \supp(\vec c_j)} x_j^i = v(\vec c_j) 				& \forall j \in \{1,\dots,m\}\\
\sum_{i \in S}\sum_{j = 1}^m x_j^i \ge \cal A^*(\CS,\vec x,S) 	& \forall S \subseteq N\nonumber\\
x_j^i \ge 0											& \forall j \in\{1,\dots,m\};\forall i \in \supp(\vec c_j)\nonumber
\end{eqnarray}
The first set of constraints ---also called {\em efficiency} constraints--- ensures that $\vec x$ is indeed a valid imputation, while the second set of constraints ---also called {\em stability} constraints--- ensures that $(\CS,\vec x)$ is in the $\cal A$-core of $\cal G$. The number of constraints in~\eqref{eq:isstable} is exponential in $n$; moreover, there is no guarantee that $\cal A^*(\CS,\vec x,S)$ is a value that is linear in $\vec x$. However, if the set of constraints described in~\eqref{eq:isstable} is linear in $\vec x$, then one can use a simple modification of the algorithm described in Theorem~\ref{thm:checkcore} as a {\em separation oracle}; given a linear program with an arbitrary number of constraints, a separation oracle is an algorithm whose input is a candidate point, and can decide in polynomial time whether the point satisfies all constraints, and can output a violated constraint otherwise. Given such an algorithm, one can decide if there exists a point that satisfies all constraints. 

We note that the set of constraints~\eqref{eq:isstable} is linear for the conservative, refined, and optimistic arbitration functions. Thus, we obtain the following corollary. 
\begin{corollary}\label{cor:isstable}
$\isstable$ is decidable in time polynomial in $n$ and $\wmax$ for all instances $\tup{\cal G,\cal A,\CS}$ such that $\cal G$ is a 2-OCF game with a tree interaction structure, and $\cal A$ is the conservative, refined or optimistic arbitration function.
\end{corollary}
Indeed, Corollary~\ref{cor:isstable} holds for {\em any} arbitration function for which the set of constraints~\eqref{eq:isstable} is linear.

\section{Beyond Tree Interactions}\label{sec:tw}
In previous sections, we have shown that if the game $\cal G$ is a discrete 2-OCF game with a tree interaction structure, then most relevant stability notions can be computed in time polynomial in $n$, the number of players, and $\wmax$, the maximal weight of any player. We now show how to extend our algorithms to 2-OCF games that have an interaction graph that is not a tree. 
\gc{I guess the following  
is some reason to consider the interaction among agents in $H$ when assessing the stability or optimality of a given coalition structure.}

If $\cal G$ is a 2-OCF game, then its interaction graph is a simply a graph with either simple edges or self edges. Moreover, if $\cal G$ is a $k$-OCF game then its interaction graph contains no edges of size $\ge k$. In this section, we show how our algorithms and their complexity can be parameterized by the {\em treewidth} of the game's interaction graph. The algorithms we describe assume that $\cal H$ is connected; however, all our results hold even if $\cal H$ is not connected, by simply applying our methods to each of the connected components of $\cal H$ separately.
We employ an important graph parameter, called the {\em treewidth}~\citep{robertson1984graph}.
Given a graph $\Gamma = \tup{N,E}$, a {\em tree decomposition} of $\Gamma$ is a tree $\cal T$ whose vertices are subsets of $N$ (we write $V(\cal T)$ to denote the vertices of $\cal T$ and $E(\cal T)$ to denote its edges), and which satisfies the following three conditions
\begin{enumerate}
\item
If $e\in E$ then there is some vertex $S \in V(\cal T)$ such that $e \subseteq S$.
\item
Given any two vertices $S,S' \in V(\cal T)$ such that there is some $i \in N$ that is in $S \cap S'$, $i$ appears in every vertex on the path between $S$ and $S'$.
\end{enumerate} 
Given a tree decomposition $\cal T$ of $\cal H$, let us write $\width(\cal T)$ to be $\max\{|S|\mid S \in V(\cal T)\} - 1$; we define the {\em treewidth} of a hypergraph $\cal H$ to be 
$$\tw(\cal H)\eqdef\min\{\width(\cal T) \mid \cal T \mbox{ is a tree decomposition of } \cal H\}.$$ 
We note that the $-1$ is simply a normalization factor, which ensures that the treewith of trees is 1; in fact, a graph is a tree if and only if its treewidth~\citep{robertson1984graph} is 1. Given a tree decomposition $\cal T$, we write $\nodes(\cal T) \eqdef \bigcup_{S \in V(\cal T)} S$; i.e. $\nodes(\cal T)$ is the set of all agents that are in the nodes of $\cal T$.

We say that a problem is {\em fixed parameter tractable} with respect to a parameter $k$ if the problem can be decided in time $f(k)n^c$ where $f$ is some function of $k$ and $c$ is a constant independent of $k$ and $n$. Intuitively, if $k$ is set to a constant, then the problem can be solved quickly.

Treewidth is often used as a parameter in the parameterized complexity analysis of graph related combinatorial problems; Courcelle's theorem~\citep{courcelle} states that any graph property that can be stated using a fairly standard set of operators (monadic second order logic) is fixed parameter tractable, with the treewidth of the graph being the parameter. Treewidth has also been used in the study of cooperative games, both for studying the computational complexity of finding solution concepts~\citep{greco2011complexity}, and in studying their structure~\citep{meir2013cost}. 

We now generalize the algorithmic results shown in previous sections to 2-OCF games whose interaction graphs have a treewidth of $k$; more specifically, we provide generalization of Theorems~\ref{thm:vtree},~\ref{thm:arbtree} and~\ref{thm:checkcore} for games whose interaction graphs have a treewidth of $k$. We note that deciding whether an interaction graph has a treewidth of $k$ (and finding a tree decomposition of the hypergraph $\cal H$ of width at most $k$) is fixed parameter tractable in $k$. Finally, we overload notation and write $\tw(\cal G)$ to be the treewidth of the interaction graph of $\cal G$, where $\cal G$ is a discrete 2-OCF game.

\begin{theorem}\label{thm:vgraph}
$\optval$ is decidable in time polynomial in $n$ and $\wmax^{\tw(\cal G)+1}$ for all  instances $\tup{\cal G,\vec c,V}$ such that $\cal G$ is a 2-OCF game.
\end{theorem}
\begin{proof}
We again show how to compute an optimal coalition structure when all agents invest all their resources; the reduction to a general coalition $\vec c$ is trivial.
Let $\cal T$ be a tree decomposition of the interaction graph of $\cal G$ such that $\width(\cal T) = k$. Let us choose some $R \in V(\cal T)$ to be the root of $\cal T$; for any $X \in V(\cal T)$, let us write $\cal T_X$ the subtree rooted in the vertex $X$, $p(X)$ to the parent of $X$ in $\cal T$, and $C_X$ to be the children of $X$ in $\cal T$. Intuitively, in order to compute an optimal coalition structure, agents in $X \cap p(X)$ need to decide how much to allocate to their own subtree $\cal T_X$, and how much to allocate to working with their parent. Let us write $\opt( \cal T_X(\vec q))$ to be the value of an optimal coalition structure over the nodes in $\cal T_X$, but with the agents in $X\cap p(X)$ investing only $\vec q$ in working with $\cal T_X$. 
We observe that 
$$\opt(\cal T_X(\vec q)) = \max\left\{v^*(\vec y+ \sum_{Y \in C_X} \vec x_Y) + \sum_{Y \in C_Y}\opt(\cal T_Y(\vec z_Y)) \right\},$$
where $\vec z_Y$ is the amount that the set $X \cap Y$ devotes to working with $\cal T_Y$, and $\vec x_Y$ is what is allocated to working with $X$; $\vec y$ is the vector of resources of $X\setminus \bigcup_{Y \in C_Y}Y$, assuming that those members of $X \cap p(X)$ contribute according to $\vec q$, i.e. $\vec y = \min\{\vec q ,\vec W^{X \setminus \bigcup_{Y \in C_Y} Y}\}$. Thus, it must hold that $\vec y+ \sum_{Y \in C_Y}\vec x_Y+\vec z_Y = \min\{\vec q,\vec W^X\}$, and $\vec x_Y+\vec z_Y \le \vec \min\{\vec q,W^{X\cap Y}\}$ for all $Y \in C_X$.

Taking a similar approach to that used in Theorem~\ref{thm:vtree}, we employ dynamic programming in order to compute $\opt(\cal T_X(\vec q))$. We write $\opt(\cal T_X(\vec q;j))$ to be the most that $\cal T_X$ can make if $X\cap p(X)$ allocates $\vec q$ to working with $\cal T_X$, and only the first $j$ children are considered, where $C_X$ is set to be $\{Y_1,\dots,Y_m\}$. $\cal T_X(\vec q;0)$ is simply $v^*(\min\{W^X,\vec q\})$, and for all $j \ge 1$: 
$$\opt(\cal T_X(\vec q;j)) = \max\left\{\opt(\cal T_X(\vec q - \vec z;j-1) + \opt(\cal T_{Y_j}(\vec z))\mid \vec z \le \min\{\vec q,\vec W^{X \cap Y_j}\}\right\}.$$
To conclude, assuming we have computed $\opt(\cal T_Y(\vec z))$ for all $Y \in C_X$ and all $\vec z$, we can compute $\opt(\cal T_X(\vec q)$ in time polynomial in $\wmax^{\tw(\cal G)+1}$ and linear in $|C_X|$, which implies that the total running time of the dynamic program is polynomial in $\wmax^{\tw(\cal G)}$ and linear in $n$.
\end{proof}

A similar approach can be used in order to compute the most that a set can get by deviating from an arbitrary graph. The same key observation used in Theorem~\ref{thm:arbtree} is made here: in order to compute the most that a set $S$ can get by deviating, we first replace $v_i^*(w)$ with $\bar{v}_i^*(w)$, where $\bar{v}_i^*(w) = \max\{\alpha_i(w-x) + v_i^*(x)\mid 0 \le x \le w\}$, and $\alpha_i(w)$ is the most that $i$ can from the arbitration function if it leaves a total of $w$ of its resources with non-deviators. Having replaced $v_i^*$ with $\bar{v}_i^*$ we run the algorithm described in Theorem~\ref{thm:vgraph} to obtain the following:
\begin{theorem}\label{thm:arbgraph}
$\arbval$ is decidable in time polynomial in $n$ and $\wmax^{\tw(\cal G)+1}$ for all instances $\tup{\cal G,\cal A,\vec c,V}$ such that $\cal A$ is local and $\cal G$ is a 2-OCF game.
\end{theorem}

Finally, we provide an algorithm for deciding instances of $\checkcore$ that runs in time polynomial in $n$ and $\wmax^{\tw(\cal G)+1}$. 
\begin{theorem}\label{thm:checkcoregraph}
An instance of $\checkcore$ is decidable in time polynomial in $n$ and $\wmax^{\tw(\cal G)+1}$ for all instances $\tup{\cal G,\cal A,(\CS,\vec x)}$ such that $\cal A$ is local and $\cal G$ is a 2-OCF game.
\end{theorem}
\begin{proof}
Given an outcome $(\CS,\vec x)$, our goal is to find a subset $S \subseteq N$ such that $e(\CS,\vec x,S)< 0$ if such a subset exists. Let $\cal T$ be the tree decomposition of the interaction graph of $\cal G$, and we again choose some $R \in V(\cal T)$ to be the root of $\cal T$. Take some $S \subseteq X$; let us denote by $\cal T_S$ the subtree rooted in $X$,but with the members of $X \setminus S$ removed from all the nodes in $\cal T_X$. We say that a subset $T$ of $N$ is rooted in $\cal T_S$ if $T \subseteq \nodes(\cal T_S)$ and $S \subseteq T$; that is, $T$ is rooted in $\cal T_S$ only if it contains $S$, as well as being contained in $\nodes(\cal T_S)$. We write $E_S(\vec q)$ to be the excess of the unhappiest subset rooted in $S$, assuming that $S$ devotes only $\vec q \le \vec W^S$ to interacting with $\cal T_S$. 

Now, suppose that we have already computed $E_T(\vec z)$ for all $T \subseteq Y$ where $Y$ is a child of $X$ and for all $\vec z \le \vec W^Y$. Let us set $C_S = \{Y_1,\dots,Y_m\}$, where $C_S = \{Y\setminus (N\setminus S)\mid Y \in C_X\}$; we write $E_S(\vec q;j)$ to be the maximal excess achievable by any subset rooted in $\cal T_S$, with all resources allocated to the first $j$ children of $S$, and assuming that $S$ allocates $\vec q$ resources to working with $\cal T_S$. Therefore, $E_S(\vec q;0)$ is $e(\CS,\vec x,S,\vec q) = \cal A^*(\CS,\vec x,S,\vec q) - p_S(\CS,\vec x)$, where $\cal A^*(\CS,\vec x,S,\vec q)$ is the most that $S$ can get by deviating when it has only $\vec q$ resources to allocate to working with non-deviators and optimize its own payoffs.

Now, when choosing how to deviate with the $j$-th child, $S$ needs to decide how much of its resources to allocate to $Y_j$, and which subset of $Y_j$ to join into to the deviation. It has already joined all members of $S \cap Y_j$, but it now needs to choose an additional subset $T \subseteq Y_j \setminus S$ to bring into the deviation, and demand resources from it in an optimal manner; in other words, 
$$E_S(\vec q;j) = \max\left\{E_S(\vec q';j-1) + v^*(\vec q - \vec q' + \vec r) + E_T(\vec W^T - \vec r)\right\},$$
where the maximization is over all $\vec q' \le \vec q$, all $T \subseteq Y_j$ and all $\vec r \le \vec W^T$. 

Thus, we can compute $E_S(\vec q)$ in time polynomial in $|C_X|$ and $(2\wmax)^{\tw(\cal G)}$, and therefore decide $\checkcore$ in polynomial time as well.
\end{proof}

\section{Linear Bottleneck Games and the Optimistic Core}\label{sec:LBG}
In this section, we deviate from the discrete setting described so far, and assume that agents have rational weights; this is simply a matter of notational ease, as all of the definitions of discrete OCF games carry through to the rational setting. The main objective of this section is to describe a class of cooperative games with overlapping coalitions that has a non-empty optimistic core, and instances of $\optval$, $\arbval$, $\checkcore$ and $\isstable$ are decidable in polynomial time when restricted to this class. Our class of OCF games is motivated by fractional combinatorial optimization scenarios. 
In the previous sections, we make no assumptions on the structure of the characteristic function, but rather use underlying agent interaction to facilitate poly-time computation. In what follows, we do not make any assumptions on agent interactions, but rather restrict our attention to a family of characteristic functions. This approach can lead to strong results, basically allowing us to efficiently decide any and all instances of OCF games in this class.  

We begin by recalling the notion of stability under the optimistic arbitration function, which will be notion of stability we study in this section. Given the optimistic arbittration function, denoted $\cal A_o$, an $\cal A_o$-profitable deviation of a set $S \subseteq N$ from an outcome $(\CS,\vec x)$ can be described by
\begin{itemize}
\item[(a)]
the list of coalitions $\CS|_S\subseteq \CS'\subseteq \CS$ that $S$ fully withdraws resources from. These are coalitions that $S$ does not wish to retain payoffs from, thus it fully withdraws its resources from them, and utilizes those resources to maximize its own profits.
\item[(b)]
the partial deviation of $S$ $\CS''$ from $\CS\setminus \CS'$, i.e. the amount of resources each $i \in S$ withdraws from each coalition in $\CS\setminus \CS'$. The coalitions in $\CS'$ are those that $S$ does wish to retain interactions with, and is thus willing to maintain the payoffs to $N\setminus S$ in those coalitions, effectively assuming the marginal cost of its deviation from those coalitions.
\end{itemize}
$S$ is then allowed to use the resources which it has withdrawn according to $\CS'$ and $\CS''$ in order to maximize its own profits, while absorbing the damage it has caused $\CS\setminus \CS'$ by withdrawing $\CS''$. 
We define a large class of OCF games that is motivated by combinatorial optimization and resource allocation scenarios,
and prove that these games always have a non-empty optimistic core. 
Moreover, we show that for games in this class an optimal coalition structure can be found
using linear programming, and the dual LP solution can be used to find an imputation in the
optimistic core. Our results in this section build on prior work on classic cooperative game
theory, where dual solutions have been used
to derive explicit payoff divisions that guarantee core stability~\citep{deng1997alg,jain2007cost,markakis2005core}; indeed, one can view our results as stating that not only are the games described in these works stable against deviations in the classic cooperative sense (i.e. have a non-empty conservative core), they are also resistant to deviations when much more lenient agent behavior is assumed. First, let us define the class of games we are interested in.

\begin{definition}
A {\em Linear Bottleneck Game} $\cal G =(N, \omega, \cal T)$
is given by a set of players $N = \{1,\dots,n\}$, a list
$\vec W = (W_1, \dots, W_n)$ of players' {\em weights},
and a list of {\em tasks}
$\cal T=(T_1, \dots, T_m)$, where each task $T_j$ is associated with a set of players $A_j\subseteq N$
who are needed to complete it, as well as a {\em value} $\pi_j\in\R_+$.
We assume that $A_j\neq A_{j'}$ for $j\neq j'$, and for each $i\in N$
there is a task $T_k\in\cal T$ with $A_k=\{i\}$.
The characteristic function of this game is defined as follows:
given a partial coalition $\vec c\in \cal W$, we set
$$
v(\vec c)=
\begin{cases}
\pi_j\cdot\min\limits_{i\in A_j}c^i &\mbox{\emph{if} $\supp(\vec c)=A_j$ \emph{for some} $j\in[m]$}\\
0	&\mbox{\emph{otherwise}}.
\end{cases}
$$
\end{definition}
These games are linear in the sense that the payoff earned by a partial coalition scales
linearly with the smallest contribution to this coalition; the smallest contribution is the ``bottleneck'' contribution, since the contribution of no member but the smallest member affects the value of the coalition.
The assumption that $A_j\neq A_{j'}$ for $j\neq j'$ ensures that the characteristic function
is well-defined; that is, each task is associated with a unique set of players that can complete it. Finally, since each player can work on his own (possibly earning a payoff
of $0$), all resources are used. This assumption will be useful when proving our results, since it allows us to invest unused agent resources in dummy tasks. 
\subsection{Some Examples}
LBGs can be used to describe a variety of settings; a more complete overview of their descriptive power can be seen in~\cite{deng1997alg}. However, for the sake of completeness, we provide three examples below. First, LBGs can describe
{\em multicommodity flow games}~\citep{vazirani2001approx,markakis2005core}. 
Briefly, in multicommodity flow games pairs of vertices
in a network want to send and receive flow, which has to be transmitted by
edges of the network. This setting can be modeled by a linear bottleneck game,
where both vertices and edges are players. The weight of an edge player
is the capacity of his edge, while the weight of a vertex player is the amount
of commodity he possesses. 

More formally, there are two types of agents in the multicommodity flow game: {\em suppliers}, denoted $N_s$, and {\em distributors}, denoted $N_d$. Given a directed graph $\Gamma$ with an edge set $E(\Gamma)$ and a node set $V(\Gamma)$, each supplier $i \in N_s$ controls a pair of nodes $(s_i,t_i)$, and has certain amount $W_i$ of a {\em commodity}, with a per-unit price of $\pi_i$. Each distributor $j \in N_d$ controls an edge $e_j$, such that $\{e_j\}_{j \in N_d} = E(\Gamma)$; each edge has a weight, or capacity, $w(e)$. A task in this game is to transfer a certain amount of a commodity owned by $i$ from $s_i$ to $t_i$. In this setting, each path from $s_i$ to $t_i$ is a task, with its associated set being the distributors on the path, and the supplier controlling $(s_i,t_i)$; the value of a coalition is the amount of the commodity supplied by $i$ that the path transfers, times the per-unit value of that commodity. We note that in this description, the number of possible tasks may be exponential in the number of agents, however, as discussed in~\citep{markakis2005core}, there do exist other, more succinct, ways of describing the problem that result in the same solution, and to which our techniques can be applied. We do maintain the current description, as it does highlight the fact that multicommodity flow games are indeed linear bottleneck games. 

Another example of a linear bottleneck game occurs in network routing settings. Consider again a directed graph $\Gamma$ with an edge set $E(\Gamma)$ and a node set $V(\Gamma)$. Here, agents are nodes, and each node $i$ has a certain weight $W_i$ (this can be thought of as processing power, or amount of memory). 
Now, the tasks in this setting are to transfer data from certain source nodes to certain target nodes, with each such $(s_j,t_j)$ pair associated with a per-unit payoff $\pi_j$. Unlike multicommodity flow games, the amount of data to be transferred is unlimited, with the only limitation on transfer power stemming from agents' own capacity constraints. 

Finally, consider a bipartite graph with the node sets $A,B$, such that $A\cap B = \emptyset$, and with no edges among the members of $A$ or the members of $B$. Agents are nodes, and for every $a \in A$ and $b \in B$, the edge $e =\{a,b\}$ has a certain value $\pi_e$. Each agent $i \in A\cup B$ has a weight $W_i$. The tasks here are the edges, with the member nodes of each edge being the set required to complete the task, and the payoff being $\pi_e$. This setting can be thought of as a slightly more generalized fractional weighted matching game, or, alternatively, as a trading market, where $A$ is a set of sellers and $B$ is a set of buyers. Having $a \in A$ and $b \in B$ form a coalition means that $b$ agrees to buy from $a$ for the set price of $\pi_e$ per unit. The total value of the coalition structure can be thought of as the total volume of exchanges made in the trading market.

\subsection{Computing Stable Outcomes in LBGs}
Before we proceed, let us make some simple observations on the structure of optimal coalition structures in LBGs.
\begin{lemma}\label{lem:lbgCS}
Given an LBG $\cal G = \tup{N,\vec W,\cal T}$, 
there is some optimal coalition structure $\CS$ such that 
\begin{enumerate}
\item[(a)]
for all $\vec c\in \CS$ we have $c^i = c^j$ for all $i,j \in \supp(\vec c)$.
\item[(b)]
$w_i(\CS) = W_i$ for all $i\in N$.
\item[(c)]
each $A_j$ forms at most one coalition in $\CS$.
\end{enumerate}
\end{lemma}
\begin{proof}
First, note that given an optimal coalition structure $\CS$ for a linear bottleneck game,
we can assume without loss of generality that for every $\vec c$
in $\CS$ and every $i,k \in A_j$ we have $c^i\omega^i = c^k\omega^k$:
investing more weight than one's team members does not increase the payoff from the task,
so a player might as well use this weight to work alone. Second, since we assume that there is a task that an agent can complete alone, the value of a coalition structure can only increase when agents invest any unused weight in working alone. 
Finally, it can be assumed that
$\CS$ contains at most one coalition $\vec c$ with $\supp(\vec c)=A_j$ for each $j=1, \dots, m$:
if $\supp(\vec c)=\supp(\vec d)=A_j$, then $v(\vec c+\vec d)\ge v(\vec c)+v(\vec d)$, so two
coalitions with the same support can be merged.
This implies that we can assume that in an optimal coalition structure
each $A_j$ forms at most one coalition $\vec c_j$. 
\end{proof}
Lemma~\ref{lem:lbgCS} implies that an optimal coalition structure can be described by a list $C_1,\dots C_m$,
indicating how much weight is allocated to each task.

We can now write a linear program that finds an optimal coalition structure
for an LBG $\cal G=\tup{N, \vec W, \cal T}$:
\begin{eqnarray}\label{eq:LBG}
\textrm{max:}	&\sum_{j = 1}^m c_j\pi_j \\
\textrm{s.t.} & \sum_{j:i \in A_j} c_j \le W_i & \forall i \in N\nonumber\\
							& c_j \ge 0 & \forall j \in [m] \nonumber
\end{eqnarray}

The dual of LP~\eqref{eq:LBG} is
\begin{eqnarray}\label{eq:LBGdual}
\textrm{min:} 	& \sum_{i = 1}^n W_i\gamma_i \\
\textrm{s.t.} 	& \sum_{i \in A_j} \gamma_i \ge \pi_j 	& \forall j \in [m] \nonumber\\
			& \gamma_i \ge 0					&\forall i \in N\nonumber
\end{eqnarray}
Let $\widehat{c_1},\dots, \widehat{c_m}$ and
$\widehat{\gamma_1},\dots,\widehat{\gamma_n}$ be optimal solutions to~\eqref{eq:LBG} and~\eqref{eq:LBGdual}
respectively.
Let $\CS$ be the coalition structure that corresponds to $\widehat{c_1},\dots, \widehat{c_m}$.
We construct a payoff vector $\vec x$ for $\CS$ as follows:
for every $j = 1, \dots, m$
we set $x_j^i = \widehat{\gamma_i}\widehat{c_j}$, if $i\in A_j$, and $x^i_j=0$ otherwise.
In words, each player $i$ has some ``bargaining power'' $\widehat{\gamma_i}$,
and is paid for each task he works on
in proportion to his bargaining power.
Note that both $\CS$ and $\vec x$ can be computed efficiently
from the description of the game.
We will now show that $\vec x$ is an imputation for $\CS$,
and, moreover, $(\CS, \vec x)$ is in the optimistic core.

\begin{theorem}\label{thm:LBGopt}
Let $\cal G=\tup{N, \vec W, \cal T}$ be a linear bottleneck game, and let $\CS$ and $\vec x$
be the coalition structure and the payoff vector constructed above.
Then $\vec x\in I(\CS)$ and
$(\CS, \vec x)$ is in the optimistic core of $\cal G$.
\end{theorem}
\begin{proof}
First, we argue that $\vec x\in I(\CS)$.
To see that $\vec x$ satisfies coalitional efficiency,
note that the sum of payoffs from task $T_j$ is
\begin{eqnarray*}
\sum_{i\in A_j} x_j^i & = & \sum_{i \in A_j} \widehat{\gamma_i}\widehat{c_j} \\
				& = & \widehat{c_j}\sum_{i\in A_j} \widehat{\gamma_i}
\end{eqnarray*}
As $\widehat{\gamma_1}, \dots, \widehat{\gamma_n}$ is an optimal solution to~\eqref{eq:LBGdual},
we have either $\sum_{i \in A_j} \widehat{\gamma_i} = \pi_j$ or $\widehat{c_j}=0$ (by complementary slackness). Thus, for any task $T_j$
that is actually executed (i.e., $\widehat{c_j}>0$), its total payoff $\pi_j\widehat{c_j}$ is
shared only by players in $A_j$.

We now show that the outcome $(\CS,\vec x)$ is in the optimistic core. 
We can assume without loss of generality that $\CS$ allocates non-zero weight to the first $k$ tasks
$T_1, \dots, T_k$ and no weight to the rest ($k\le m$). Consider a deviation from $(\CS, \vec x)$ by a set $S$. 
This deviation can be described by a list of tasks that $S$ abandons completely, and the amount
of weight that players in $S$ withdraw from all other tasks. Assume without loss of generality
that the tasks that $S$ abandons completely are $T_{\ell+1}, \dots, T_k$
(this list includes all tasks $T_j$ with $A_j\subseteq S$), and for each
$j=1, \dots, \ell$ each member of $A_j\cap S$ withdraws $z_j$ units of weight from $T_j$. Observe that ``non-uniform'' deviations are no better than ``uniform'' ones, i.e. if one agent withdraws more weight from a coalition than the rest in an optimal deviation, then the rest may as well withdraw the same weight.

When players in $S$ deviate, they lose their payoff from $T_{\ell+1}, \dots, T_k$,
and their payoff from $T_1, \dots, T_\ell$ is reduced by $\sum_{j=1}^\ell z_j\pi_j$.

For each $i\in S$, set
$\nu_i=\sum_{\ell<j\le k, i\in A_j} \widehat{c_j}$,
$Z_i = \sum_{j\le \ell, i\in A_j} z_j$:
$\nu_i$ is the total amount of weight that $i$ withdraws from tasks $T_{\ell+1}, \dots, T_k$
while $Z_i$ is the total amount that $i$ withdraws from $T_1, \dots, T_\ell$.
The profit that $S$ obtains from optimally using the withdrawn resources
is given by the following linear program:
\begin{eqnarray}\label{eq:LBG-S}
\textrm{max:} & \sum_{A_j\subseteq S} c_j\pi_j\\
\textrm{s.t.} & \sum_{j:i \in A_j, A_j\subseteq S} c_j \le \nu_i+ Z_i & \forall i \in S\nonumber
\end{eqnarray}
The dual of LP \eqref{eq:LBG-S} is
\begin{eqnarray}\label{eq:LBGdual-S}
\textrm{min:} & \sum_{i\in S} \gamma_i(\nu_i +Z_i)\\
\textrm{s.t.} & \sum_{i \in A_j} \gamma_i \ge \pi_j & \forall A_j\subseteq S \nonumber
\end{eqnarray}

Let $\alpha$ be the value of~\eqref{eq:LBG-S} (and hence also of~\eqref{eq:LBGdual-S}). Note that the total profit that $S$ gets by deviating equals $\alpha - \sum_{j = 1}^\ell z_j\pi_j$, where $\sum_{j = 1}^\ell z_j\pi_j$ is the total marginal loss incurred by $S$ partially deviating from $T_1,\dots,T_\ell$.
Any optimal solution to~\eqref{eq:LBGdual} is a feasible
solution to~\eqref{eq:LBGdual-S} (when looking at the restriction of the solution to those members of $i$); the constraints in~\eqref{eq:LBGdual} are more restricted than those in~\eqref{eq:LBGdual-S}, since $\nu_i+Z_i \le W_i$ for all $i \in S$. Hence, given the optimal solution of the Dual~\eqref{eq:LBGdual} restricted to $S$, $(\widehat{\gamma}_i)_{i \in S}$, we have that 
$$\alpha\le \sum_{i\in S}\widehat{\gamma_i}(\nu_i+Z_i).$$

Now, $\sum_{i\in S}\widehat{\gamma^i}\nu^i$ is exactly the payoff that $S$
was getting from $T_{\ell+1}, \dots, T_k$ under $(\CS,\vec x)$.
Further,
\begin{eqnarray*}
\sum_{i\in S}\widehat{\gamma_i}Z_i 	& = &\sum_{j=1}^\ell\sum_{i\in S\cap A_j}\widehat{\gamma_i}z_j\\
							& = & \sum_{j=1}^\ell z_j\left(\sum_{i\in S\cap A_j}\widehat{\gamma_i}\right)\le \sum_{j=1}^\ell z_j\left(\sum_{i\in  A_j}\widehat{\gamma_i}\right) \le \sum_{j=1}^\ell z_j\pi_j
\end{eqnarray*}
where the last inequality
holds as per the constraints in~\eqref{eq:LBGdual}; as previously mentioned, the latter expression
is the marginal loss that $S$ pays for withdrawing resources from $T_1, \dots, T_\ell$.
Thus, the total payoff that $S$ gets from deviating is at most $\sum_{i \in S} \widehat{\gamma_i}\nu_i$; but:
\begin{eqnarray*}
\sum_{i \in S} \widehat{\gamma_i}\nu_i 	& = & \sum_{i \in S}\widehat{\gamma_i}\sum_{\ell<j\le k, i\in A_j} \widehat{c_j}\\
								& \le &\sum_{i \in S}\widehat{\gamma_i}\sum_{j:i\in A_j} \widehat{c_j}\\
								& = & \sum_{i \in S}p_i(\CS,\vec x) = p_S(\CS,\vec x)
\end{eqnarray*}
To conclude, the total payoff $S$ receives from deviating under the optimistic arbitration function does not exceed its payoff in $(\CS,\vec x)$.
As this holds for any deviation and any $S$, $(\CS, \vec x)$ is in the optimistic core.
\end{proof}
Finding an optimal solution for a linear program and its dual can be done in polynomial time; therefore, we obtain the following immediate corollary.
\begin{corollary}
Given an LBG $\cal G = \tup{N,\vec W,\cal T}$, $\optval$, $\arbval$, and $\checkcore$ can be decided in polynomial time under the optimistic arbitration function.
\end{corollary}
Moreover, since we know that the optimistic core is not empty for LBGs, it is in particular not empty for the conservative, sensitive and refined arbitration functions.

\section{Conclusions and Future Work}
In this paper, we analyzed computational aspects of finding optimal coalition structures and stable outcomes in OCF games. The first part of our work assumed a certain structure on agent interaction: agents were only allowed to form small-sized coalitions, and computational efficiency was highly dependent on a tree-like interaction graph. In the latter part of this work, we assumed that agent interaction was not limited, but restricted ourselves to a class of characteristic functions. 

Using structural limitations as a method for ensuring tractability is interesting, but it forces us to make strong assumptions on the way agents interact; one could argue that the assumptions we require are rather unrealistic. Apart from their intrinsic interest, we believe that our results show how hard it is to efficiently compute core allocations in OCF games when one makes no assumptions on the characteristic function. 

We believe that results in the spirit of Section~\ref{sec:LBG} would prove to be more instructive; that is, finding classes of OCF games that can be stabilized with respect to certain arbitration functions would prove to be useful. 

Arbitration functions naturally encode agent behavior towards one another; more lenient reactions to deviation signify agents that are more tolerant towards each other. This tolerance may arise due to agents being completely myopic in their behavior ---as is captured by the idea of a local arbitration function--- and is an important assumption in deciding the computational complexity of stability related problems in OCF games. 
We believe that non-myopic reaction to deviation (e.g. the sensitive arbitration function) would severely hinder poly-time computability. It would be interesting to identify a meaningful class of OCF games for which computing a core allocation is possible in polynomial time if the arbitration function is local, but is not when the arbitration function is more complex. Even results related to, say, the refined and the sensitive core would be interesting. We mention that linear bottleneck games are not appropriate for this type of analysis; since the optimistic core of these games is not empty (and an optimistic core outcome can be found in polynomial time), an outcome in the sensitive core can trivially be found (simply pick the outcome computed for the optimistic core), as the optimistic core is contained in the sensitive core.

\subsection{Related Work}
Our work expands and builds upon two previous papers on OCF games. The first is the seminal work by~\cite{ocfgeb}, and the second is a recently published paper by~\cite{zick2014jair}. \cite{ocfgeb} define the OCF model and discuss some initial computational results; for example, they show that it is possible to find an outcome in the conservative core of threshold task games if agent weights are not too large; they also show that finding an outcome in the refined core is computationally harder than finding an outcome in the conservative core. This is the first indication that different arbitration functions not only lead to different outcomes, but can also raise computational barriers. While \cite{zick2014jair} do not study computational aspects of OCF games, their work does present us with some possibly useful tools in the computational analysis of OCF games. First, \cite{zick2014jair} show that certain classes of OCF games are guaranteed to have a non-empty core; however, their proofs rely on balancedness conditions, and are not computational in nature. Thus, even though some games are guaranteed to be stable with respect to some arbitration functions, computing stable outcomes may be hard.

Many works focus on computing optimal coalition structures in cooperative games. The {\em optimal coalition structure generation problem} has received plenty of attention in classic cooperative game theory literature (see \cite{sandholm99coalition,larson2000anytime,michalak2008optimal,rahwan2009anytime,rahwan2012anytime}, as well as the overview chapters in~\citep{compcoopbook,masbook}). Some authors have studied games with overlapping coalitions as well. We mention the seminal work by \cite{shehorykraus}, as well as the work by~\cite{ocfalgo07} and \cite{ocfsearch10}. While these works are highly related to ours, their methodology and objectives are different. We are not only interested in forming an optimal coalition structure, but also in stable revenue division.

Computational aspects of classic cooperative games have been an object of extensive study. Preliminary computational results can be attributed to the founders of the field; \cite{mann60sv} and \cite{mann62sv} study methods to compute the Shapley value~\cite{shapleyvalue} (an important solution concept in cooperative game theory) using exact methods and via approximation. Their results, while not phrased in the language of modern computational complexity, are of a computational nature. The seminal paper by \cite{deng1994complexity} was the precursor of several works on the subject. A non-comprehensive list includes \cite{deng1997alg,ieong2005mcnets,matsui2000survey,elkind2007complexity,greco2011complexity}; see the survey in~\cite{compcoopbook}. The methods used by \cite{brafman2010planning} are quite similar to ours, and echo the ideas of \cite{Demange2004}, in exploiting the tree structure of games for computational purposes. Finally, we mention some recent works on stability in OCF environments. \cite{zhang2013overlapping} employ the OCF model in analyzing wireless networks, and show some stability results in this setting; \cite{acker2012fair} study a pairwise collaboration model that is similar to the OCF model studied here. Their paper studies pairwise equilibria in this model, rather than stability. 
\section*{References}

\bibliographystyle{elsarticle-num-names}      
\bibliography{bib}   

\end{document}